\documentclass[a4paper, 12pt]{article}
\usepackage[english]{babel}
\usepackage{amsmath, amsfonts, amsthm, mathrsfs, latexsym}
\usepackage[dvips]{graphicx}

\newcommand{\abs}[1]{\lvert#1\rvert}
\newcommand{\norm}[1]{\lVert#1\rVert}
\newcommand{\field}[2]{\mathbb{#1}^{#2}}
\newcommand{\deriv}[2]{\partial_{#1}^{#2}}
\newcommand{\hilbert}[1]{\mathscr{#1}}
\newcommand{\fock}{\hilbert{F}}
\newcommand{\fockfin}{\fock_{\mathrm{fin}}}
\newcommand{\coinf}{C_{0}^{\infty}}
\newcommand{\dgamma}{\mathrm{d\Gamma}}

\newcommand{\E}{{\mathrm{e}}}
\newcommand{\I}{\mathrm{i}}
\newcommand{\B}{\hilbert{B}}
\newcommand{\Or}{{\mathcal{O}}}
\newcommand{\epsi}{\varepsilon}
\newcommand{\hepsi}{H^{\epsi}}
\newcommand{\hc}{H^{c}}
\newcommand{\hlambda}{H_{\lambda}}
\newcommand{\pop}{\hat{p}}
\newcommand{\Hp}{H_{\mathrm{p}}}

\newcommand{\hdres}{H_{\mathrm{dres}}^{\epsi, \sigma}}
\newcommand{\heff}{H_{\mathrm{eff}}^{\epsi}}
\newcommand{\hdrestwo}{H_{\mathrm{dres}}^{(2)}}
\newcommand{\hdiag}{H_{\mathrm{D}}^{(2)}}
\newcommand{\hdiagp}{H_{\mathrm{D}, \mathrm{p}}^{(2)}}
\newcommand{\hdiagtilde}{\tilde{H}_{\mathrm{D}}^{(2)}}
\newcommand{\id}{\mathbf{1}}
\newcommand{\hepsisigma}{H^{\epsi, \sigma}}
\newcommand{\hiepsi}[1]{\hat{h}_{#1}^{\epsi}}
\newcommand{\hilambda}[1]{h_{#1}}
\newcommand{\hp}{\hat{h}_{\mathrm{p}}^{\epsi}}
\newcommand{\hiepsidres}[1]{\hat{h}_{#1}^{\mathrm{dres}}}
\newcommand{\hiepsisigma}[1]{\hat{h}_{#1}^{\epsi, \sigma}}
\newcommand{\hdarw}{H_{\mathrm{darw}}}
\newcommand{\vdarw}{V_{\mathrm{darw}}}
\newcommand{\vspin}{V_{\mathrm{spin}}}
\newcommand{\opw}{\mathrm{Op}^{W}_{\epsi}}
\newcommand{\omegap}{\omega_{\mathrm{p}}}

\newcommand{\uepsi}{\hilbert{U}_{\epsi}}
\newcommand{\uepsisigma}{\hilbert{U}_{\epsi, \sigma}}

\newcommand{\lf}{{\sf L}_{\mathrm{f}}}
\newcommand{\llambda}{{\sf L}_{\lambda}}
\newcommand{\lp}{{\sf L}_{\mathrm{p}}}
\newcommand{\lone}{{\sf L}_{1}}
\newcommand{\lfour}{{\sf L}_{4/3}}

\newcommand{\hilbertp}{\hilbert{H}_{\mathrm{p}}}

\newcommand{\vcoul}{V_{\varphi\,\mathrm{coul}}}
\newcommand{\vcoulsigma}{V_{\varphi_{\sigma}\,\mathrm{coul}}}
\newcommand{\Hf}{H_{\mathrm{f}}}
\newcommand{\normomega}[1]{\norm{#1}_{\omega}}
\newcommand{\proj}[2]{\hat{\pi}^{#1}_{#2}}
\newcommand{\pepsim}{P_{M}^{\epsi}}
\newcommand{\pzeroel}{P_{0}^{\mathrm{E}, \mathrm{L}}}
\newcommand{\qm}{Q_{M}}
\newcommand{\qless}[1]{Q_{\leq\mathrm{#1}}}
\newcommand{\qmore}[1]{Q_{>\mathrm{#1}}}
\newcommand{\phijsigma}{\Phi_{j, \sigma}}
\newcommand{\phijsigmal}{\Phi_{j, \sigma}^{\mathrm{L}}}
\newcommand{\phijzerol}{\Phi_{j, 0}^{\mathrm{L}}}
\newcommand{\uonelchi}{U^{(1)}_{\mathrm{L}, \chi}}
\newcommand{\uonelchistar}{U^{(1)\,*}_{\mathrm{L}, \chi}}
\newcommand{\uonelsigma}{U_{1,\sigma}^{\mathrm{L}}}
\newcommand{\uonelsigmastar}{U_{1,\sigma}^{\mathrm{L}\,*}}
\newcommand{\eqex}{\overset{!}{=}}
\newcommand{\tr}[1]{\mathrm{Tr}_{#1}}
\newcommand{\ztwo}{\mathbb{Z}_{2}}
\newcommand{\omegaf}{\Omega_{\mathrm{F}}}
\newcommand{\omegar}{\omega_{\mathrm{R}}}

\newtheorem{theorem}{Theorem}
\newtheorem*{theorem*}{Theorem}
\newtheorem{lemma}{Lemma}
\newtheorem*{lemma*}{Lemma}
\newtheorem*{prop}{Proposition}
\newtheorem{proposition}{Proposition}
\newtheorem{corollary}{Corollary}

\theoremstyle{remark}
\newtheorem{remark}{Remark}

\title{Quasi-static Limits in Nonrelativistic Quantum Electrodynamics}
\author{L. Tenuta}
\date{}

\begin{document}

\maketitle
\begin{center}
Mathematisches Institut, Eberhard-Karls-Universit\"at, Auf der Morgenstelle 10, 72076, T\"ubingen, Germany.

e-mail: lucattilio.tenuta@uni-tuebingen.de
\end{center}

\begin{abstract}
We consider a system of $N$ nonrelativistic particles of spin $1/2$ interacting with the quantized Maxwell field (mass zero and spin one) in the limit when the particles have a small velocity.

Two ways to implement the limit are considered: $c\to\infty$ with the velocity $v$ of the particles fixed, the case for which rigorous results have already been discussed in the literature, and $v\to 0$ with $c$ fixed. The second case can be rephrased as the limit of heavy particles, $m_{j}\to \epsi^{-2}m_{j}$, observed over a long time, $t\to\epsi^{-1}t$, $\epsi\to 0^{+}$, with kinetic energy $E_{\mathrm{kin}} = \Or(1)$.

Focusing on the second approach we construct subspaces which are invariant for the dynamics up to terms of order $\epsi \sqrt{\log(\epsi^{-1})}$ and describe effective dynamics, \emph{for the particles only}, inside them. At the lowest order the particles interact through Coulomb potentials. At the second one, $\epsi^{2}$, the mass gets a correction of electromagnetic origin and a velocity dependent interaction, the Darwin term, appears.

Moreover, we calculate the radiated piece of the wave function, i.\ e., the piece which leaks out of the almost invariant subspaces and calculate the corresponding radiated energy.
\end{abstract}

\section{Introduction}
A system of nonrelativistic particles of spin $1/2$ interacting with the quantized radiation field is described by the so-called Pauli-Fierz Hamiltonian, or ``nonrelativistic quantum electrodynamics''. The model is thought to have an extremely wide range of validity, apart from phenomena connected to gravitational forces and from other ones typical of high-energy physics like pair creation, whose description requires the use of full relativistic QED.

This belief is mainly based on the analysis of some formal limit cases, which can be accurately studied both from a theoretical and an experimental point of view. Indeed, the interaction between charged particles is usually described by instantaneous pair potentials of Coulomb-type, without introducing the field as dynamical variable. This is known to be a good approximation if the particles move sufficiently slowly. One aim of this paper is a mathematically rigorous justification of this fact, i. e., the derivation of the Schr\"odinger equation with Coulomb potentials, and the second order velocity dependent corrections to them, starting from nonrelativistic quantum electrodynamics. In addition, a formula is provided for the wave function of the radiated photons and the corresponding radiated energy, which is the quantum equivalent of the Larmor formula of classical electrodynamics.

In more detail the model considered is given, excluding the addition of the electronic spin, by the canonical quantization of a system of $N$ classical charges interacting through the Maxwell field. 

A sharp ultraviolet cutoff is introduced assuming that each charge has a charge distribution given by
\begin{equation}\label{formfactor}
\varrho_{j}(x) = e_{j}\varphi(x), \quad x\in\field{R}{3},
\end{equation}
where the form factor satisfies $\hat{\varphi}(k)=(2\pi)^{-3/2}$ for $\abs{k}\leq\Lambda$, $0$ otherwise (note that there is \emph{no infrared cutoff}).

The classical equations of motion are given by
\begin{equation}\label{classicaleqfield}\begin{split}
& \frac{1}{c}\deriv{t}{}B(x, t) = -\nabla\times E(x, t),\\
& \frac{1}{c}\deriv{t}{}E(x, t) = \nabla\times B(x, t) - \sum_{j=1}^{N}e_{j}\varphi\big(x-q_{j}(t)\big)\frac{\dot{q}_{j}(t)}{c},
\end{split}
\end{equation}
with the constraints
\begin{equation}
\nabla\cdot E(x, t) = \sum_{j=1}^{N}e_{j}\varphi\big(x - q_{j}(t)\big),\qquad \nabla\cdot B(x, t)=0,
\end{equation}
and the Newton equations for the particles,
\begin{equation}\label{classicaleqparticles}
m_{l}\ddot{q}_{l}(t) = e_{l}[E_{\varphi}(q_{l}(t), t)) + \frac{\dot{q}_{l}(t)}{c}\times B_{\varphi}(q_{l}(t), t)], \quad l=1, \ldots N,
\end{equation}
where $E_{\varphi}(x, t) := (E*_{x}\varphi)(x, t)$ and analogously for $B_{\varphi}$.

The canonical quantization of this system in the Coulomb gauge is described, e.\ g., in (\cite{Sp}, chapter $13$). The Hilbert space of the pure states is given by
\begin{equation}
\hilbert{H} := \hilbertp\otimes \fock.
\end{equation}
The space for the particles, $\hilbertp$, is defined by\footnote[1]{The formalism presented holds also in the case when all the particles are equal and their Hilbert space is given by the subspace of totally antisymmetric wave functions. In this case, the dipole radiation given in \eqref{psirad} is zero. We consider therefore the general case of different particles.}
\begin{equation}
\hilbertp := L^{2}(\field{R}{3}\times\ztwo)^{\otimes N},
\end{equation}
where $\field{R}{3}$ is the configuration space of a single particle and $\ztwo$ represents its spin.

The state space for a single photon is $L^{2}(\field{R}{3}\times\ztwo)$, where $\field{R}{3}$ is the momentum space of the photon and $\ztwo$ represents its two independent physical helicities. The photon Fock space is therefore
\begin{equation}\label{fockspace}
\fock := \oplus_{M=0}^{\infty}\otimes_{(s)}^{M}L^{2}(\field{R}{3}\times\ztwo),
\end{equation}
where $\otimes_{(s)}^{M}$ denotes the $M$-symmetric tensor product and $\otimes_{(s)}^{0}L^{2}(\field{R}{3}\times\ztwo) := \field{C}{}$. We denote by $\omegaf$ the vector $(1, 0, \ldots)$, called the Fock vacuum.

The dynamics of the system are generated by the Hamiltonian
\begin{equation}\label{hc}
\hc := \sum_{j=1}^{N}\frac{1}{2m_{j}}\bigg[\sigma_{j}\cdot\bigg(-\I\nabla-\frac{1}{\sqrt{c}}e_{j}A_{\varphi}(x_{j})\bigg)\bigg]^{2} + \vcoul(x) + c\Hf,
\end{equation}
where all the operators appearing are independent of $c$ and we use units in which $\hbar=1$. $\sigma_{j}$ is a vector whose components are the Pauli matrices of the $j$th particle, $A_{\varphi}(x_{j})$ denotes the quantized transverse vector potential in the Coulomb gauge, $\vcoul$ is the smeared Coulomb potential and $\Hf$ the free field Hamiltonian. The reader who is not familiar with the notation is advised to look at section \ref{description}, where the model is described in more detail.

To implement practically the idea that the particles move ``slowly'', a standard procedure, applied also in classical electrodynamics (see, e.g., \cite{Ja}, \cite{LaLi}), is to take the limit $c\to\infty$\footnote[2]{In the classical case, a more refined and precise analysis is carried through in \cite{KuSp1}, \cite{KuSp2}. The authors consider, loosely speaking, initial conditions which represent free particles moving together with the field they generate (``dressed'' particles or charge solitons), with a velocity of order $\Or(\epsi^{1/2})$ with respect to the speed of light. Assuming that the particles are at time $t=0$ far apart (relative distance of order $\Or(\epsi^{-1})$) and rescaling suitably the dynamical variables, they show that the particles remain at a relative distance of order $\Or(\epsi^{-1})$ for long times (of order $\Or(\epsi^{-3/2})$) and on this time scale their motion is governed by effective dynamics. The possibility to implement an analogous limit in the quantum case is unclear, because there is no obvious quantum counterpart to the classical charge solitons. The Pauli-Fierz Hamiltonian without infrared cutoff has indeed no ground state in the Fock space for fixed total momentum different from zero \cite{Fr} \cite{Ch}. We stick therefore to the more pragmatic choice $c\to\infty$.

The states which we define through the dressing transformation $\hilbert{U}_{\epsi}$ in equation \eqref{pepsim} should be considered approximate dressed states valid for small velocities of the particles.}. Since $c$ is a quantity with a dimension, one should actually say that $\abs{v}/c\to 0$, where $v$ is a typical velocity of the particles. This can be achieved in two ways, fixing $v$ and letting $c\to\infty$ or fixing $c$ and letting $v\to 0$.

In the classical case this is reflected in the fact that the limit $c\to\infty$ is equivalent, up to a rescaling of time, to the limit of heavy particles, as one can easily verify replacing in equations \eqref{classicaleqfield}-\eqref{classicaleqparticles} $m_{l}$ with $\epsi^{-2}m_{l}$, $t$ with $\epsi^{-1}t$, and looking at the limit $\epsi\to 0$.   

We will show that in the quantum case the two procedures are non equivalent anymore, a fact that can be intuitively explained by the presence of an additional scale given by $\hbar$.

In this paper we concentrate on the limit of heavy particles observed over a long time. An additional aim is to point out similarities and differences between the two limits in the quantum context and to compare the results we get for the Pauli-Fierz model with the ones valid for the Nelson model, where the particles and the photons are spinless, \cite{TeTe}.

We recall briefly in the next subsection some known results about the case $c\to\infty$ and then illustrate in more detail the limit $\epsi\to 0$.

\subsection{The limit $c\to\infty$}

This case, as observed by Spohn \cite{Sp}, has the form of a weak coupling limit, fact which was already noted for the Nelson model by Davies \cite{Da2}, who formulated also a general scheme to analyze the limit dynamics in the weak coupling case \cite{Da1} (an extended notion of weak coupling limit for Pauli-Fierz systems in which the Hilbert space of the particles is finite dimensional has been examined in \cite{DeDe}). 

Davies looks loosely speaking at the limit $\lambda\to 0$ for the time evolution generated by an Hamiltonian of the form
\begin{equation*}
\lambda^{-2}(H_{0} + \lambda H_{\mathrm{int}}),
\end{equation*}
which corresponds physically to a weak interaction, whose effect is however observed over the long time scale $\lambda^{-2}$.

The Hamiltonian $\hc$ assumes a similar form if we consider it on the long time scale defined by $c^{2}$. Putting
\begin{equation*}
\lambda := c^{-3/2}
\end{equation*}
we get indeed
\begin{equation}\label{equationhlambda}
c^{2}\hc = \lambda^{-2}(\Hf + \lambda^{2/3}\Hp + \lambda\hilambda{1} + \lambda^{4/3}\hilambda{4/3}):=\lambda^{-2}\hlambda,
\end{equation}
where
\begin{eqnarray}
\Hp &:=& -\sum_{j=1}^{N} \frac{1}{2m_{j}}\Delta_{j} + \vcoul,\\
\hilambda{1} &:=& \sum_{j=1}^{N} \frac{e_{j}}{m_{j}}\I\nabla_{x_{j}}\cdot A_{\varphi}(x_{j}) - \frac{e_{j}}{2m_{j}}\sigma_{j}\cdot B_{\varphi}(x_{j}),\\
\hilambda{4/3} &:=& \sum_{j=1}^{N}\frac{e_{j}^{2}}{2m_{j}}:A_{\varphi}(x_{j})^{2}: \label{hlambdafour}\quad ,
\end{eqnarray}
where $B_{\varphi} := \nabla\times A_{\varphi}$ and we normal order the quadratic term.

Applying Davies scheme one gets in the end (\cite{Sp}, theorem $20.5$)
\begin{theorem}\label{ctoinfinity}
Let $\psi\in H^{1}(\field{R}{3N}, \field{C}{2^{N}})$, then
\begin{equation}
\lim_{c\to\infty}\norm{(\E^{-\I\hc c^{2}t} - \E^{-\I \hdarw c^{2}t})\psi\otimes\omegaf}_{\hilbert{H}}=0,
\end{equation}
where
\begin{eqnarray}
\hdarw &:=& \Hp + c^{-2}\vdarw + c^{-2}\vspin,\\
\vdarw &:=& -\sum_{j, l=1}^{N}\frac{e_{j}e_{l}}{m_{j}m_{l}}\int_{\field{R}{3}}dk\, \frac{\abs{\hat{\varphi}(k)}^{2}}{2\abs{k}^{2}}\E^{\I k\cdot x_{j}}\I\nabla_{x_{j}}\cdot(\id - \kappa\otimes\kappa)\I\nabla_{x_{l}}\E^{-\I k\cdot x_{l}}, \label{vdarwin}\\
\vspin &:=& -\sum_{j, l=1}^{N}\frac{e_{j}e_{l}}{12m_{l}m_{j}}\sigma_{j}\cdot\sigma_{l}\int_{\field{R}{3}}dk\, \abs{\hat{\varphi}(k)}^{2}\E^{\I k\cdot(x_{j} - x_{l})},
\end{eqnarray}
where $(\kappa\otimes\kappa)_{ij}:=\kappa_{i}\kappa_{j}$ and $\kappa := k/\abs{k}$.
\end{theorem}
$\vdarw$ gives rise to a correction of electromagnetic origin to the mass of the particles  and to a velocity dependent potential, called the Darwin term. It already appears in classical electrodynamics when the dynamics of the particles are expanded up to terms of order $(v/c)^{2}$ (see, e. g., \cite{Ja} or \cite{LaLi}).

For the convenience of the reader and to ease the comparison with the results for the limit of heavy masses we give a formal derivation of this theorem in appendix \ref{appendix}. 

We note here that the method employed in the weak coupling case forces one to consider as initial condition for the field just the Fock vacuum, which contains no photons at all. There is therefore no analogy with the physical picture that every particle should be described by a ``dressed state'', loosely speaking the particle itself dragging with it a cloud of ``virtual'' photons.

\subsection{The limit $m\to\infty$}

The situation is different in the case $m_{j}\to\infty$, which is more conveniently studied adopting units where $c=1$. Replacing $m_{j}$ by $\epsi^{-2}m_{j}$ we get then the Hamiltonian
\begin{equation}\label{hepsi}\begin{split}
\hepsi :=& \sum_{j=1}^{N} \frac{1}{2m_{j}}\pop_{j}^{2} + \vcoul + \Hf  - \epsi\frac{e_{j}}{m_{j}}\pop_{j}\cdot A_{\varphi}(x_{j}) - \epsi^{2}\frac{e_{j}}{2m_{j}}\sigma_{j}\cdot B_{\varphi}(x_{j}) + \\
&+ \epsi^{2}\frac{e_{j}^{2}}{2m_{j}}:A_{\varphi}(x_{j})^{2}:\quad ,
\end{split}
\end{equation}
where we have indicated with $\pop_{j}$ the $\epsi$-momentum of the $j$th particle
\begin{equation}
\pop_{j} := -\I\epsi\nabla_{x_{j}}.
\end{equation}

As already pointed out talking about the classical case, the dynamics have to be observed over times of order $\Or(\epsi^{-1})$. This is necessary in order to see non trivial effects, because we consider initial states with bounded kinetic energy. Since the particles have a mass of order $\Or(\epsi^{-2})$ this means that their velocity is in the original time scale of order $\Or(\epsi)$.

To analyze the limit $\epsi\to 0$ we construct a unitary dressing transformation $\hilbert{U}_{\epsi} : \hilbert{H}\to\hilbert{H}$, which allows us to define dressed states for small velocities of the particles and to introduce a clear notion of real and virtual photons. More precisely, in the new representation defined by $\hilbert{U}_{\epsi}$, the vacuum sector $\hilbertp\otimes\omegaf$ corresponds to states of dressed particles without real photons, while in the original Hilbert space a state with $M$ real photons is a linear combination of states of the form
\begin{equation*}
\hilbert{U}_{\epsi}^{-1}(\psi\otimes a(f_{1})^{*}\cdots a(f_{M})^{*}\omegaf), \quad\textrm{where}\quad \psi\in\hilbertp,\, f_{1}, \ldots, f_{M}\in L^{2}(\field{R}{3}\times\ztwo).
\end{equation*}
The projector on the subspace corresponding to dressed particles with $M$ real photons is therefore
\begin{equation*}
\pepsim := \hilbert{U}_{\epsi}^{*}(\id_{\mathrm{p}}\otimes\qm)\hilbert{U}_{\epsi},
\end{equation*}
where $\qm$ denotes the projector on the $M$-particles subspace of the Fock space. 

In short, we will show that the subspaces $\pepsim\hilbert{H}$ are approximately invariant for the dynamics defined by $\hepsi$ on times of order $\Or(\epsi^{-1})$. Moreover, on this time scale we will give effective dynamics for states inside such a subspace, with an error of order $\Or(\epsi^{2}\log(\epsi^{-1}))$. The effective dynamics contain the Darwin correction described in \eqref{vdarwin}, but \emph{no spin dependent term}. One can get an idea of why this happens comparing the expression of $\hepsi$ with that of $\hlambda$, equation \eqref{hepsi} and \eqref{equationhlambda}. In $\hepsi$ the spin dependent term is of second order, while in $\hlambda$ is of the first one. In the limit $\epsi\to 0$ the analogue of $\vspin$ would be of order $\Or(\epsi^{4})$, therefore it does not appear in an expansion of the time evolution till second order. Finally we compute the leading order part of the state which makes a transition between $P_{0}^{\epsi}$ and $P_{1}^{\epsi}$, which corresponds to the emission of one real photon. The corresponding radiated energy is given by a quantum analogue of the Larmor formula.

The procedure to construct the unitary $\hilbert{U}_{\epsi}$ is explained in detail in \cite{TeTe} for the Nelson model. The technique used is based on space-adiabatic perturbation theory \cite{Te}, a method which allows to expand the dynamics generated by a pseudodifferential operator with an $\epsi$-dependent semiclassical symbol. 

The main difficulty in all models concerning the interaction of particles with a quantized field of zero mass is that, because of soft photons, the principal symbol of the Hamiltonian has no spectral gap, which is a condition required to apply the methods of \cite{Te}. In the case of $\hepsi$ we have indeed  
\begin{equation}\label{principalsymbol}
h_{0}(p, q) := \sum_{j=1}^{N} \frac{1}{2m_{j}}p_{j}^{2} + \vcoul(q) + \Hf,\quad (p, q)\in\field{R}{3N}\times\field{R}{3N} \quad .
\end{equation}
For every fixed $(p, q)$ this is an operator on $\fock$ and has a ground state given by $\omegaf$, at the threshold of the continuous spectrum. The corresponding eigenvalue
\begin{equation}
E_{0}(p, q) = \sum_{j=1}^{N} \frac{1}{2m_{j}}p_{j}^{2} + \vcoul(q)
\end{equation}
is the symbol of an Hamiltonian acting just on $\hilbertp$ and describing the particles interacting through the smeared Coulomb potential.

The trouble connected to the absence of the spectral gap is solved by introducing an effective gap, considering the Hamiltonian $\hepsisigma$ where the form factor $\hat{\varphi}$ (see equation \eqref{formfactor}) is replaced by $\hat{\varphi}_{\sigma}(k) := (2\pi)^{-3/2}$ for $\sigma<\abs{k}<\Lambda$, $0$ otherwise. 
\begin{prop}(see proposition \ref{infraredcutoff})

Suppose that the cutoff $\sigma$ is a function of $\epsi$, $\sigma = \sigma(\epsi)$, such that $\sigma(\epsi)<\epsi^{2}$, then
\begin{equation}
\norm{\E^{-\I t \hepsi/\epsi} - \E^{-\I t \hepsisigma/\epsi}}_{\mathcal{L}(\hilbert{H}_{0}, \hilbert{H})} \leq C\abs{t}\sigma(\epsi)^{1/2}
\end{equation}
\end{prop}
where 
\begin{equation}
\hilbert{H}_{0} := D\big(\hepsi_{0})
\end{equation}
is the domain of the free Hamiltonian
\begin{equation}
\hepsi_{0} := \sum_{j=1}^{N} \frac{1}{2m_{j}}\pop_{j}^{2} + \Hf 
\end{equation}
with the corresponding graph norm.

Fixing $\sigma$, e. g., as a sufficiently high power of $\epsi$ we can then replace the original dynamics with infrared cutoff ones. 

For $\hepsisigma$ it is possible to build a dressing operator $\hilbert{U}_{\epsi, \sigma}$ which can be expanded in a series of powers of $\epsi$ with $\sigma$-dependent coefficients which are at most logarithmically divergent. Using it we define the dressed Hamiltonian
\begin{equation}
\hdres := \uepsisigma\hepsisigma\uepsisigma^{*}
\end{equation}
which can be expanded in a series of powers of $\epsi$ in $\mathcal{L}(\hilbert{H}_{0}, \hilbert{H})$, with coefficients which are also at most logarithmically divergent in $\sigma$. The different coefficients in the expansion correspond to different physical effects which can be now clearly separated according to their order of magnitude in $\epsi$. 

The first result we find, as we already mentioned above, is that the dressed $M$-photons subspaces are approximately invariant for the dynamics:
\begin{theorem*}(see corollary \ref{adibaticinvariancepepsim}).

Given a $\chi\in\coinf(\field{R}{})$ and a function $\sigma(\epsi)$ such that
\begin{equation}\label{conditionssigma}
\epsi^{-2}\sigma(\epsi)^{1/2}\to 0,\qquad \epsi\sqrt{\log(\sigma(\epsi)^{-1})}\to 0, \qquad \epsi\to 0^{+},
\end{equation}
then
\begin{equation}
\norm{[\E^{-\I\hepsi\frac{t}{\epsi}}, \pepsim]\chi(\hepsi)}_{\mathcal{L}(\hilbert{H})} = \Or\big(\sqrt{M+1}\abs{t}\epsi\sqrt{\log(\sigma(\epsi)^{-1})}\big)\quad ,
\end{equation}
where
\begin{equation}\label{pepsim}
\pepsim := \hilbert{U}^{*}_{\epsi, \sigma(\epsi)}(\id_{\mathrm{p}}\otimes\qm)\hilbert{U}_{\epsi, \sigma(\epsi)} \quad .
\end{equation}
\end{theorem*}
The adiabatic decoupling which guarantees the invariance of the subspaces holds uniformly only on states in which the particles have a uniformly bounded kinetic energy. For this reason we introduce a cutoff function on the total energy $\chi$, which gives rise automatically to a bounded kinetic energy for the slow particles. 

In the following we assume that the function $\sigma(\epsi)$ has been fixed so that \eqref{conditionssigma} is satisfied. One can then approximate the dynamics of the particles inside each almost invariant subspace.
\begin{theorem*}(see theorem \ref{densitymatrix}).

Let $S$ be a bounded observable for the particles, $S\in\mathcal{L}(\hilbertp)$, and $\omega \in \hilbert{I}_{1}(P_{M}^{\epsi}\chi(\hepsi)\hilbert{H})$ a density matrix for a mixed dressed state with $M$ free photons, whose time evolution is defined by
\begin{equation*}
\omega(t):= \E^{-\I t\hepsi/\epsi}\omega \E^{\I t\hepsi/\epsi}\quad .
\end{equation*}

We have then
\begin{equation*}\begin{split}  \tr{\hilbert{H}}\bigg(\big(S\otimes\id_{\fock}\big)\omega(t)\bigg)& = \tr{\hilbertp}\bigg(S \E^{-\I t\heff}\tr{\fock}(\omega) \E^{\I t \heff}\bigg) +\\ &+ \Or(\epsi^{3/2}\abs{t})(1-\delta_{M0}) + \Or\big(\epsi^{2}\log(\sigma(\epsi)^{-1})(\abs{t} + \abs{t}^{2})\big),
\end{split}
\end{equation*}
where $\delta_{M0} = 1$, when $M=0$, $0$ otherwise, and 
\begin{equation}\begin{split}
\heff &:= \sum_{j=1}^{N} \frac{1}{2m_{j}}\pop_{j}^{2} + \vcoul + \\
&-\epsi^{2}\sum_{l, j=1}^{N}\frac{e_{j}e_{l}}{m_{j}m_{l}}\int_{\field{R}{3}}dk\, \frac{\abs{\hat{\varphi}(k)}^{2}}{2\abs{k}^{2}}\E^{\I k\cdot x_{j}}\pop_{j}\cdot(\id - \kappa\otimes\kappa)\pop_{l}\E^{-\I k\cdot x_{l}} = \\
&= \sum_{j=1}^{N} \frac{1}{2m_{j}}\pop_{j}^{2} + \vcoul + \epsi^{2}\vdarw\quad .
\end{split}
\end{equation}
\end{theorem*}

\begin{remark}
Even though the subspaces $\pepsim$ depend on the choice of the infrared cutoff, the effective Hamiltonian is infrared regular and therefore \emph{independent} of $\sigma$. Moreover, as we briefly mentioned above, it contains the corrections to the mass of the particles and the Darwin term, but no spin dependent term (compare with theorem \ref{ctoinfinity}). This topic is further discussed in the proof of theorem \ref{firstordertimeevolution} and in remark \ref{remarkspin}.
\end{remark}

Since the subspaces $\pepsim$ are only approximately invariant, there is a piece of the wave function which ``leaks out'' in the orthogonal complement. This correspond physically to the emission or absorption of free photons. For a system starting in the dressed vacuum the leading order of the wave function of the emitted photon is given in the next theorem.

\begin{theorem*}(see corollary \ref{radiatedpiece}).

Up to terms of order $\Or\big(\epsi^{2}\log(\sigma(\epsi)^{-1})(\abs{t} + \abs{t}^{2})\big)$, the radiated piece for a system starting in the dressed vacuum ($M=0$) is given by
\begin{equation}\label{psirad}\begin{split}
&\Psi_{\mathrm{rad}}(t) := (\id - P_{0}^{\epsi})\E^{-\I\frac{t}{\epsi}\hepsi}P_{0}^{\epsi}\chi(\hepsi)\Psi \cong\\ & \cong -\E^{-\I t \hiepsi{0}}\frac{\I\epsi}{\sqrt{2}} \frac{\hat{\varphi}_{\sigma(\epsi)}(k)}{\abs{k}^{3/2}}e_{\lambda}(k)\cdot\int_{0}^{t}ds\, \E^{\I(s-t)\abs{k}/\epsi}\opw\big(\ddot{D}(s; x, p)\big)\psi(x)\, ,
\end{split}
\end{equation}
where $e_{\lambda}(k)$ is the polarization vector of a photon with helicity $\lambda$,
\begin{equation}
\hiepsi{0}:= \sum_{j=1}^{N} \frac{1}{2m_{j}}\pop_{j}^{2} + \vcoul + \Hf,
\end{equation}
\begin{equation}\label{psi}
\psi(x) := <\omegaf, \chi(\hiepsi{0})\Psi>_{\fock}\, \in\hilbertp,
\end{equation}
\begin{equation}
D(s; x, p) := \sum_{j=1}^{N}\frac{e_{j}}{m_{j}}x_{j}^{cl}(s; x, p),
\end{equation}
$\opw$ denotes the Weyl quantization acting on a suitable symbol space on $\field{R}{3N}\times\field{R}{3N}$ and $x_{j}^{cl}$ is the solution to the classical equations of motion
\begin{equation*}\begin{split}
& m_{j}\ddot{x}_{j}^{cl}(s; x, p) = -\nabla_{x_{j}}\vcoul(x^{cl}(s; x, p)),\\
& x_{j}^{cl}(0; x, p)= x_{j},\qquad \dot{x}_{j}^{cl}(0; x, p)= p_{j}m_{j}^{-1},\quad j=1, \ldots, N\, .
\end{split}
\end{equation*}
\end{theorem*}

\begin{remark}
As explained in detail in remark \ref{remarkradiatedpiece}, generically the norm of the radiated piece is bounded below by $\Or\big(\epsi\log(\epsi\sigma(\epsi)^{-1})\big)$, which means that the subspace $P_{0}^{\epsi}$ is near optimal, i.\ e.\ the transitions are at least of order $\Or\big(\epsi\log(\epsi\sigma(\epsi)^{-1})\big)$.

Note that, like in classical electrodynamics, when all the particles are equal, the leading order of the radiated piece vanishes, because $D$ is then proportional to the position of the center of mass, whose acceleration is zero. 
\end{remark}

\begin{remark}
Even though the radiated wave function has no limit when $\epsi\to 0$, because $\varphi(k)\abs{k}^{-3/2}\notin L^{2}(\field{R}{3})$, the corresponding radiated energy has a limit. Defining
\begin{equation}\label{radiatedenergy}
E_{\mathrm{rad}}(t) := \langle\Psi_{\mathrm{rad}}(t), \Hf\Psi_{\mathrm{rad}}(t)\rangle,
\end{equation}
we get to the leading order (see remark \ref{remradiatedpower})
\begin{equation}
P_{\mathrm{rad}}(t) := \frac{d}{dt}E_{\mathrm{rad}}(t) \cong \frac{\epsi^{3}}{3\pi^{2}}\langle \psi, \opw\big(\abs{\ddot{D}(t)}^{2}\big)\psi\rangle_{\hilbertp}\quad .
\end{equation}
\end{remark}

In the case of the Nelson model analogous results are proved in \cite{TeTe}, which contains also a detailed discussion of the adiabatic framework. The form of the effective dynamics is equal, the only difference, as one can expect, is in the radiated piece, which contains here explicitly the helicity of the photon. Another difference is that the principal symbol of the Pauli-Fierz Hamiltonian, defined in \eqref{principalsymbol}, is diagonal with respect to the Fock projectors $\qm$, while for the Nelson Hamiltonian one needs a dressing transformation already at the leading order. This makes the analysis of the Pauli-Fierz case somewhat less technical.

The effective dynamics for $M=0$ (dressed vacuum) was calculated by Spohn (\cite{Sp}, section $20.2$) in the case when the photon has a small mass, $m_{\mathrm{ph}}>0$, which introduces a gap in the principal symbol of the Hamiltonian. He however states that these effective dynamics are identical with the ones calculated for the case $c\to\infty$, while we have already remarked that the spin dependent term cannot be present when $\epsi\to 0$. In the case $m_{\mathrm{ph}}>0$, moreover, the transitions between the almost invariant subspaces become smaller than any power of $\epsi$ and therefore it is not known how to get an explicit expression for them.

The Pauli-Fierz Hamiltonian has also been extensively studied to get informations about its spectral and scattering structure. Not pretending to be exhaustive, we refer the reader interested to these aspects to \cite{BFS}, \cite{DeGe}, \cite{FGS}, \cite{GLL} and references therein.

In section \ref{description} we complete the description of the model and discuss the approximation of the original dynamics through infrared cutoff ones. In section \ref{unitary} the construction of the dressing operator $\hilbert{U}$ is discussed, and applied in section \ref{dressed} to the study of the dressed Hamiltonian. The main results on the effective dynamics and the radiated piece are contained in section \ref{effective}. Finally, appendix \ref{appendix} contains a sketch of the proof of theorem \ref{ctoinfinity}, for the case $c\to\infty$.

\section{Preliminary facts}\label{description}
In this section we elaborate on the definition of the Pauli-Fierz model and discuss some preliminary facts like the self-adjointness of the Hamiltonian and the approximation of the original dynamics through infrared cutoff ones.

\subsection{Fock space and field operator}

(The proofs of the statements we claim can be found in (\cite{ReSi2}, section X.7)).

We denote by $\fockfin$ the subspace of the Fock space, defined in \eqref{fockspace}, for which $\Psi^{(M)}=0$ for all but finitely many $M$. Given $f\in L^{2}(\field{R}{3}\times\ztwo)$, one defines on $\fockfin$ the annihilation operator by
\begin{equation}\begin{split}
(a(f)\Psi)^{(M)}(k_{1}, \lambda_{1}; \ldots; k_{M}, \lambda_{M}):= & \sqrt{M+1}\sum_{\lambda=1}^{2}\int_{\field{R}{3}}dk\, f(k, \lambda)^{*}\cdot\\
&\cdot\Psi^{(M+1)}(k, \lambda; k_{1}, \lambda_{1} \ldots, k_{M}, \lambda_{M}) \, .
\end{split}
\end{equation}
The adjoint of $a(f)$ is called the creation operator, and its domain contains $\fockfin$. On this subspace they satisfy the canonical commutation relations
\begin{equation}\begin{split}
& [a(f), a(g)^{*}]=\langle f, g\rangle_{L^{2}(\field{R}{3}\times\ztwo)},\\
& [a(f), a(g)]=0, \quad [a(f)^{*}, a(g)^{*}]=0 \, .
\end{split}
\end{equation}
Since the commutator between $a(f)$ and $a(f)^{*}$ is bounded, it follows that $a(f)$ can be extended to a closed operator on the same domain of $a(f)^{*}$.

On this domain one defines the Segal field operator
\begin{equation}
\Phi(f) := \frac{1}{\sqrt{2}}(a(f) + a(f)^{*})
\end{equation}
which is essentially self-adjoint on $\fockfin$. Moreover, $\fockfin$ is a set of analytic vectors for $\Phi(f)$. From the canonical commutation relations it follows that
\begin{equation}\label{commphi}
[\Phi(f), \Phi(g)] = \I\Im\langle f, g\rangle_{L^{2}(\field{R}{3}\times\ztwo)}\, .
\end{equation}

Given a self-adjoint multiplication operator by the function $\omega$ on the domain $D(\omega)\subset L^{2}(\field{R}{3})$, we define
\begin{equation}
\fock_{\omega, \mathrm{fin}} := \mathcal{L}\{\omegaf, a(f_{1})^{*}\cdots a(f_{M})^{*}\omegaf: M\in\field{N}{}, f_{j}\in D(\omega)\otimes\field{C}{2}, j=1, \dots, M\},
\end{equation}
where $\mathcal{L}$ means ``finite linear combinations of''.

On $\fock_{\omega, \mathrm{fin}}$ we define the second quantization of $\omega$, $\dgamma(\omega)$, by
\begin{eqnarray*}
(\dgamma(\omega)\Psi)^{(M)}(k_{1}, \lambda_{1}; \ldots; k_{M}, \lambda_{M}) &:=& \sum_{j=1}^{M}\omega(k_{j})\Psi^{(M)}(k_{1}, \lambda_{1}; \ldots; k_{M}, \lambda_{M}),\\
\dgamma(\omega)\omegaf &:=& 0,
\end{eqnarray*}
which is essentially self-adjoint. In particular, the free field Hamiltonian $\Hf$ acts as
\begin{eqnarray*}
(\Hf\Psi)^{(M)}(k_{1}, \lambda_{1}; \ldots; k_{M}, \lambda_{M}) &=& \sum_{j=1}^{M}\abs{k_{j}}\Psi^{(M)}(k_{1}, \lambda_{1}; \ldots; k_{M}, \lambda_{M}),\\
\Hf\omegaf = 0,
\end{eqnarray*}
and is self-adjoint on its maximal domain.

From the previous definitions, given $f\in D(\omega)\otimes\field{C}{2}$, one gets the commutation properties
\begin{equation}\label{commutators}\begin{split}
& [\dgamma(\omega), a(f)^{*}] = a(\omega f)^{*}, \quad [\dgamma(\omega), a(f)] = -a(\omega f),\\
& [\dgamma(\omega), \I\Phi(f)] = \Phi(\I\omega f) \, .
                \end{split}
\end{equation}

\subsection{The Pauli-Fierz model}

Using the Segal field operator one can write the quantized vector potential and the magnetic field appearing in \eqref{hc} as 
\begin{eqnarray}
A_{\varphi}(x) &=& \Phi(v_{x}),\\ 
v_{x}(k, \lambda) &:=& f(k, \lambda)\E^{-\I k\cdot x}, \quad f(k, \lambda):= \frac{e_{\lambda}(k)}{\abs{k}^{1/2}}\hat{\varphi}(k),\\ 
B_{\varphi}(x) &=& \nabla_{x}\times A_{\varphi}(x) = -\Phi(\I k\times v_{x}) \label{magneticfield} ,
\end{eqnarray}
where $e_{\lambda}(k)$, $\lambda = 1, 2$, are, for simplicity, real photon polarization vectors satisfying
\begin{equation}
e_{\lambda}(k)\cdot e_{\mu}(k) = \delta_{\lambda\mu},\qquad k\cdot e_{\lambda}(k)=0\qquad .
\end{equation}

The smeared Coulomb potential is given by
\begin{equation}\label{coulomb}
\vcoul(x) = \frac{1}{2}\sum_{j, l=1}^{N}e_{j}e_{l}\int_{\field{R}{3}}dk\, \E^{\I k\cdot(x_{j}-x_{l})}\frac{\abs{\hat{\varphi}(k)}^{2}}{\abs{k}^{2}} \quad .
\end{equation}
Analogous expressions hold for the infrared cutoff Hamiltonian $\hepsisigma$, where the form factor $\hat{\varphi}$ is replaced by $\hat{\varphi}_{\sigma}$.

To separate more clearly the terms of different order in the Hamiltonian $\hepsi$, equation \eqref{hepsi}, it is useful to write it as
\begin{equation}\label{expansionhepsi}
\hepsi = \sum_{i=0}^{2}\epsi^{i}\hiepsi{i}
\end{equation}
where 
\begin{eqnarray}
\hiepsi{0}&:=& \sum_{j=1}^{N} \frac{1}{2m_{j}}\pop_{j}^{2} + \vcoul + \Hf, \nonumber\\
\hiepsi{1}&:=& - \sum_{j=1}^{N}\frac{e_{j}}{m_{j}}\pop_{j}\cdot \Phi(v_{x_{j}}),\\
\hiepsi{2}&:=& \sum_{j=1}^{N}\bigg(\frac{e_{j}}{2m_{j}}\sigma_{j}\cdot \Phi(\I k\times v_{x_{j}}) + \frac{e_{j}^{2}}{2m_{j}}:\Phi(v_{x_{j}})^{2}:\bigg)\quad.
\end{eqnarray}
Each of the $\hiepsi{i}$ is of order $\Or(1)$ when applied to functions of bounded kinetic energy. The coefficients for $\hepsisigma$ will be denoted by $\hiepsisigma{i}$.

As proved by Hiroshima \cite{Hi} using functional integral techniques the Hamiltonian $\hepsi$ (and analogously $\hc$) is self-adjoint on $\hilbert{H}_{0}$ for every value of the masses, charges and number of particles. Since however we study the limit $\epsi \to 0$ (respectively $c\to\infty$) it is enough for our purposes to show this using Kato theorem, like, e. g., in \cite{BFS}.

Even though the proof is well known, we repeat it because we need to show that the graph norms which appear are equivalent uniformly in $\epsi$ and $\sigma$. Moreover, the estimates which appear in the proof will be useful in propositions \ref{infraredcutoff} and lemma \ref{lemmachiinfra}.

Given $f\in L^{2}(\field{R}{3}\times\ztwo)$, we define  
\begin{equation}
\normomega{f} := (\norm{f\abs{k}^{-1/2}}_{L^{2}(\field{R}{3}\times\ztwo)}^{2} + \norm{f}^{2}_{L^{2}(\field{R}{3}\times\ztwo)})^{1/2} \quad .
\end{equation}
One has then the \emph{basic estimate} 
\begin{proposition}\label{basicestimate}
\begin{equation}
\norm{a^{\sharp}(f_{1})\cdots a^{\sharp}(f_{n})(\Hf + \id)^{-n/2}}_{\mathcal{L}(\fock)}\leq C_{n}\normomega{f_{1}}\cdots\normomega{f_{n}}\quad ,
\end{equation}
where $a^{\sharp}(f)$ can be $a(f)$ or $a^{*}(f)$.
\end{proposition}

\begin{proposition}\label{selfadjointness}
Both Hamiltonians $\hepsi$ and $\hepsisigma$ are self-adjoint on $\hilbert{H}_{0}$. Moreover the graph norms they define are equivalent to the one defined by $\hepsi_{0}$ uniformly in $\epsi$ and $\sigma$. The same holds for the graph norm defined by $(\hepsi)^{1/2}$ and $(\hepsisigma)^{1/2}$.
\end{proposition}

\begin{proof} \emph{(We give the proof for $\hepsisigma$, the one for $\hepsi$ is the same)}

The regularized Coulomb potential is a bounded function, therefore for it the statement is trivial.

We choose a vector $\Psi$ in a core of $\hepsi_{0}$ made up of smooth functions with compact support both in $x$ and $k$. 

For the term of order $\epsi$ we get then
\begin{equation*}\begin{split}
& \bigg\lVert- \epsi\sum_{j=1}^{N}\frac{e_{j}}{m_{j}}\pop_{j}\cdot A_{\varphi_{\sigma}}(x_{j})\Psi\bigg\rVert_{\hilbert{H}}=\bigg\lVert- \epsi\sum_{j=1}^{N}\frac{e_{j}}{m_{j}}\Phi(v_{x_{j},\sigma})\cdot\pop_{j} \Psi\bigg\rVert_{\hilbert{H}}\leq\\
& \leq\epsi\sum_{j=1}^{N}\frac{\abs{e_{j}}}{m_{j}}\sum_{\alpha=1}^{3}\norm{\Phi(v_{x_{j},\sigma}^{\alpha})\pop_{j}^{\alpha}\Psi}\leq \epsi\sum_{j=1}^{N}\frac{\abs{e_{j}}}{m_{j}}\sum_{\alpha=1}^{3}\norm{\Phi(v_{x_{j},\sigma}^{\alpha})(\Hf+\id)^{-1/2}}_{\mathcal{L}(\hilbert{H})}\cdot\\
&\cdot \norm{(\Hf +\id)^{1/2}\pop_{j}^{\alpha}\Psi}_{\hilbert{H}}\leq\\
& \leq C\epsi\sum_{j=1}^{N}\frac{\abs{e_{j}}}{m_{j}}\normomega{v_{x_{j}, \sigma}}\norm{\Psi}_{\hilbert{H}_{0}},
                 \end{split}
\end{equation*}
so for $\epsi$ sufficiently small this term is Kato small with respect to the free Hamiltonian, with a constant uniformly bounded in $\epsi$ and $\sigma$.

An analogous estimate holds for the term with the magnetic field. For the remaining one we have
\begin{equation*}\begin{split}
& \bigg\lVert\sum_{j=1}^{N}\epsi^{2}\frac{e_{j}^{2}}{2m_{j}}:A_{\varphi_{\sigma}}(x_{j})^{2}:\Psi\bigg\rVert_{\hilbert{H}}=\bigg\lVert\epsi^{2}\sum_{j=1}^{N}\frac{e_{j}^{2}}{2m_{j}}:\Phi(v_{x_{j},\sigma})^{2}:\Psi\bigg\rVert_{\hilbert{H}}\leq\\
&\leq C\epsi^{2}\sum_{j=1}^{N}\frac{e_{j}^{2}}{2m_{j}}\normomega{v_{x_{j},\sigma}}^{2}\norm{(\Hf+\id)\Psi}_{\hilbert{H}},
\end{split}
\end{equation*}
which completes the proof.
\end{proof}

\begin{proposition}\label{infraredcutoff} If $\sigma(\epsi)<\epsi^{2}$ then
\begin{equation}\label{cutofftimeev}
\norm{\E^{-\I t \hepsi/\epsi} - \E^{-\I t \hepsisigma/\epsi}}_{\mathcal{L}(\hilbert{H}_{0}, \hilbert{H})} \leq C\abs{t}\sigma^{1/2}
\end{equation}
\end{proposition}

\begin{proof}
From the previous proposition we know that both Hamiltonians are self-adjoint on $\hilbert{H}_{0}$, so, given $\Psi\in\hilbert{H}_{0}$, we can apply Duhamel formula to get
\begin{equation*}\begin{split}
& \norm{(\E^{-\I t \hepsi/\epsi} - \E^{-\I t \hepsisigma/\epsi})\Psi}_{\hilbert{H}}\leq\frac{1}{\epsi}\int_{0}^{t}ds\, \norm{(\hepsi - \hepsisigma)\E^{-\I s\hepsi/\epsi}\Psi}_{\hilbert{H}} \quad .
\end{split}
\end{equation*}

Putting $\Psi_{s}:= \E^{-\I s\hepsi/\epsi}\Psi$,  the difference of the two Hamiltonians is
\begin{equation}\label{estimatehepsisigma}\begin{split}
(\hepsi - \hepsisigma)\Psi_{s} &= (\vcoul - \vcoulsigma)\Psi_{s} -\epsi\sum_{j=1}^{N}\frac{e_{j}}{m_{j}}\pop_{j}\cdot\Phi(\id_{(0, \sigma)}(k)v_{x_{j}})\Psi_{s}+\\
&- \epsi^{2}\sum_{j=1}^{N}\frac{e_{j}}{2m_{j}}\sigma_{j}\cdot \Phi(-\I k\times \id_{(0, \sigma)}(k)v_{x_{j}})\Psi_{s} +\\ 
&+ \epsi^{2}\sum_{j=1}^{N}\frac{e_{j}^{2}}{2m_{j}}[A_{\varphi}(x_{j})^{2} - A_{\varphi_{\sigma}}(x_{j})^{2}]\Psi_{s} \quad.
\end{split}
\end{equation}

Using the explicit expression \eqref{coulomb}, the term with the Coulomb potential gives
\begin{equation*}\begin{split}
\abs{\vcoul(x) - \vcoulsigma(x)} &\leq \frac{1}{2}\sum_{j, l=1}^{N}\abs{e_{j}e_{l}}\int \frac{dk}{\abs{k}^{2}} \bigg\lvert\abs{\hat{\varphi}}^{2} - \abs{\hat{\varphi}_{\sigma}}^{2}\bigg\rvert=\\
&=\frac{1}{12\pi^{2}}\sum_{j,l=1}^{N}\abs{e_{j}e_{l}}\sigma
\end{split}
\end{equation*}
\begin{equation}\label{estimatecoulomb}
\Rightarrow \norm{\vcoul - \vcoulsigma}_{\mathcal{L}(\hilbert{H})}=\Or(\sigma)\quad.
\end{equation}
For the term of order $\epsi$, proceeding as in the proof of proposition \ref{selfadjointness} we get that
\begin{equation*}
\bigg\lVert-\epsi\sum_{j=1}^{N}\frac{e_{j}}{m_{j}}\pop_{j}\cdot\Phi(\id_{(0, \sigma)}(k)v_{x_{j}})\Psi_{s}\bigg\rVert_{\hilbert{H}}\leq \tilde{C}\epsi\sum_{j=1}^{N}\frac{\abs{e_{j}}}{m_{j}}\normomega{\id_{(0, \sigma)}(k)v_{x_{j}, \sigma}}\norm{\Psi_{s}}_{\hilbert{H}_{0}} \quad .
\end{equation*}

From the same proposition it follows that the graph norm associated to $\hepsi_{0}$ and the one associated to $\hepsi$ are equivalent uniformly in $\epsi$ and $\sigma$, therefore
\begin{equation}\label{estimateepsisigma}\begin{split}
& \bigg\lVert-\epsi\sum_{j=1}^{N}\frac{e_{j}}{m_{j}}\pop_{j}\cdot\Phi(\id_{(0, \sigma)}(k)v_{x_{j}})\Psi_{s}\bigg\rVert_{\hilbert{H}}\leq C\epsi\sum_{j=1}^{N}\frac{\abs{e_{j}}}{m_{j}}\normomega{\id_{(0, \sigma)}(k)v_{x_{j}, \sigma}}\norm{\Psi}_{\hilbert{H}_{0}} = \\
&=\Or(\epsi\sigma^{1/2})\norm{\Psi}_{\hilbert{H}_{0}} \quad .
\end{split}
\end{equation}

The same reasoning holds for the term containing the spin, which has however a $\abs{k}$ more, which gives in the end
\begin{equation}\label{infraredspin}
\bigg\lVert- \epsi^{2}\sum_{j=1}^{N}\frac{e_{j}}{2m_{j}}\sigma_{j}\cdot \Phi(-\I k\times \id_{(0, \sigma)}(k)v_{x_{j}})\Psi_{s}\bigg\rVert_{\hilbert{H}}= \Or(\epsi^{2}\sigma^{3/2})\norm{\Psi}_{\hilbert{H}_{0}}\quad .
\end{equation}

Concerning the last term we have
\begin{equation*}\begin{split}
& \epsi^{2}\sum_{j=1}^{N}\frac{e_{j}^{2}}{2m_{j}}[A_{\varphi}(x_{j})^{2} - A_{\varphi_{\sigma}}(x_{j})^{2}]\Psi_{s}=\epsi^{2}\sum_{j=1}^{N}\frac{e_{j}^{2}}{2m_{j}}[\Phi(v_{x_{j}})^{2} - \Phi(v_{x_{j}, \sigma})^{2}]\Psi_{s}=\\
&=\epsi^{2}\sum_{j=1}^{N}\frac{e_{j}^{2}}{2m_{j}}\{[\Phi(v_{x_{j},\sigma}) + \Phi(\id_{(0, \sigma)}(k)v_{x_{j}})]^{2} - \Phi(v_{x_{j}, \sigma})^{2}\}\Psi_{s}=\\
&=\epsi^{2}\sum_{j=1}^{N}\frac{e_{j}^{2}}{2m_{j}}[\Phi(v_{x_{j}, \sigma})\Phi(\id_{(0, \sigma)}(k)v_{x_{j}}) + \Phi(\id_{(0, \sigma)}(k)v_{x_{j}})\Phi(v_{x_{j}, \sigma}) +\\
&+ \Phi(\id_{(0, \sigma)}(k)v_{x_{j}})^{2}]\Psi_{s}\quad .
\end{split}
\end{equation*}

Using again the basic estimate in proposition \ref{basicestimate} we get, for example,
\begin{equation}\label{infraredasquare}\begin{split}
& \norm{\Phi(v_{x_{j}, \sigma})\Phi(\id_{(0, \sigma)}(k)v_{x_{j}})\Psi_{s}}\leq C\normomega{v_{x_{j},\sigma}}\normomega{\id_{(0, \sigma)}(k)v_{x_{j}}}\norm{\Psi_{s}}_{\hilbert{H}_{0}}=\\
&=\Or(\sigma^{1/2})\norm{\Psi}_{\hilbert{H}_{0}},
\end{split}
\end{equation}
by the same reasoning we used for the terms of order $\epsi$.
\end{proof}

\begin{lemma}\label{lemmachiinfra}
Given a function $\chi\in\coinf(\field{R}{})$ and assuming $\sigma(\epsi)<\epsi^{2}$, then
\begin{equation}\label{cutoffchi}
\norm{\chi(\hepsi) - \chi(\hepsisigma)}_{\mathcal{L}(\hilbert{H})}\leq C\epsi\sigma^{1/2}
\end{equation}
\end{lemma}

\begin{proof}
Using the Hellfer-Sj\"ostrand formula (see, e. g., \cite{DiSj} chapter $8$), given a self-adjoint operator $A$, we can write 
\begin{equation}\label{hellfersjostrand}
\chi(A) = \frac{1}{\pi}\int_{\field{R}{2}}dx dy\ \bar{\partial}\chi^{a}(z)(A-z)^{-1}, \quad z:=x+\I y,
\end{equation}
where $\chi^{a}\in \coinf(\field{C}{})$ is an almost analytic extension of $\chi$, which satisfies the properties
\begin{eqnarray*}
\forall \bar{N}\in \field{N}{}\quad \exists\,D_{\bar{N}}&:& \abs{\bar{\partial}\chi^{a}}\leq D_{\bar{N}}\abs{\Im z}^{\bar{N}},\\
\chi^{a}_{|_{\field{R}{}}}=\chi \quad .
\end{eqnarray*}
(For the explicit construction of such a $\chi^{a}$ see \cite{DiSj}).

Applied to our case \eqref{hellfersjostrand} yields
\begin{equation*}
\chi(\hepsi) - \chi(\hepsisigma) = \frac{1}{\pi}\int_{\field{R}{2}}dxdy\, \bar{\partial}\chi^{a}(z)\big[(\hepsi-z)^{-1} - (\hepsisigma-z)^{-1}\big].
\end{equation*}
Since both Hamiltonians are self-adjoint on $\hilbert{H}_{0}$ we have 
\begin{eqnarray}\label{differenceresolvents}
 (\hepsi-z)^{-1} - (\hepsisigma-z)^{-1} &=& (\hepsisigma-z)^{-1}(\hepsisigma - \hepsi)(\hepsi-z)^{-1}\,,
 \end{eqnarray}
 and hence
 \begin{eqnarray*}\lefteqn{
 \norm{\chi(\hepsi) - \chi(\hepsisigma)}_{\mathcal{L}(\hilbert{H})} \leq}\\
& \leq&\frac{1}{\pi} \int_{\field{R}{2}} dxdy\, \abs{\bar{\partial}\chi^{a}(z)}\norm{(\hepsisigma-z)^{-1}}_{\mathcal{L}(\hilbert{H})}\norm{(\hepsisigma - \hepsi)(\hepsi-z)^{-1}}_{\mathcal{L}(\hilbert{H})}.
\end{eqnarray*}
In addition we have that
\begin{equation}\label{resolvbounded}
\norm{(\hepsisigma-z)^{-1}}_{\mathcal{L}(\hilbert{H})} \leq \frac{C}{\abs{\Im z}}.
\end{equation}
This follows because $\hilbert{H}_{0}$ is dense in the domain of $H^{\epsi=0, \sigma=0}$, and for every $\Psi\in\hilbert{H}_{0}$
\begin{equation*}
\hepsisigma\Psi \to H^{\epsi=0, \sigma=0}\Psi \quad\mbox{as}\quad  (\epsi, \sigma)\to (0, 0) \, .
\end{equation*}
According to theorem VIII.25 \cite{ReSi1}, this implies that 
\begin{equation*}
(\hepsisigma-z)^{-1}\Psi \to (H^{\epsi=0, \sigma=0}-z)^{-1}\Psi,
\end{equation*}
therefore $\abs{\Im z}\norm{(\hepsisigma-z)^{-1}\Psi}$ is bounded for every $\Psi$ and the uniform boundedness principle gives \eqref{resolvbounded}.

For the second norm we find that for $z\in \mathrm{supp}\,\chi^{a}$
\begin{eqnarray*}\lefteqn{\hspace{-1cm}
 \norm{(\hepsisigma - \hepsi)(\hepsi-z)^{-1}}_{\mathcal{L}(\hilbert{H})}\leq}\\&\leq &\norm{(\hepsisigma - \hepsi)}_{\mathcal{L}(\hilbert{H}_{0}, \hilbert{H})}\cdot\norm{(\hepsi-z)^{-1}}_{\mathcal{L}(\hilbert{H}, \hilbert{H}_{0})}\\
& \leq& \frac{C}{\abs{\Im z}}\norm{(\hepsisigma - \hepsi)}_{\mathcal{L}(\hilbert{H}_{0}, \hilbert{H})}\quad.  
\end{eqnarray*}

The right-hand side was already estimated in the previous proposition (see equation \eqref{estimatehepsisigma} and the following calculations), the only difference being that here we have to multiply the final result by $\epsi$.
\end{proof}

\section{Construction of the dressing operator}\label{unitary}

In this section we construct the unitary dressing operator $\uepsisigma$. We apply to the Pauli-Fierz Hamiltonian the general procedure explained in some detail in \cite{TeTe} and, for the case with spectral gap, in \cite{Te}. All the calculations expounded in section \ref{unitaryformal} are formal, and they serve as a guide for the rigorous definition of $\uepsisigma$ given in section \ref{unitaryrigorous}.

\subsection{The formal procedure}\label{unitaryformal}

The main idea is to build an approximate projector, $\proj{(1)}{}$, which satisfies formally
\begin{equation*}
(\proj{(1)}{})^{2} - \proj{(1)}{} = \Or(\epsi^{2}),\qquad\qquad [\proj{(1)}{}, \hepsi] = \Or(\epsi^{2}).
\end{equation*} 
Integrating over time the second equation one gets in a loose sense that $[\E^{-\I t\hepsi/\epsi}, \proj{(1)}{}] = \Or(\epsi\abs{t})$.

The projector $\proj{(1)}{}$ is found using an iterative procedure, which assumes that one can expand it in powers of $\epsi$,
\begin{equation*}
\proj{(1)}{} = \proj{}{0} + \epsi\proj{}{1},
\end{equation*}
where the coefficient $\proj{}{0}$ is a known input and must commute with the coefficient of order zero in the expansion of the Hamiltonian $\hepsi$, see equation \eqref{expansionhepsi}. As it turns out, the procedure does not work directly for $\hepsi$, but only for the infrared cutoff Hamiltonian $\hepsisigma$.

An obvious choice for $\proj{}{0}$ is 
\begin{equation*}
\proj{}{0} = \qm,
\end{equation*}
which satisfies $[\hiepsisigma{0}, \qm] = 0$.

Proceeding now in the same way as described in \cite{TeTe} we get a formal expression for the first order almost projection given by
\begin{eqnarray}\label{almostproj}
\proj{(1)}{M}&:=& \qm + \epsi\proj{M}{1},\nonumber\\
\proj{M}{1}&:=& \bigg[\qm, \I\sum_{j=1}^{N}\frac{e_{j}}{m_{j}}\pop_{j}\cdot\Phi\bigg(\frac{\I v_{x_{j},\sigma}(\lambda, k)}{\abs{k}}\bigg)\bigg] \quad.\nonumber
\end{eqnarray}

For brevity, we put from now on
\begin{equation}
\Phi_{j, \sigma}:= \Phi\bigg(\frac{\I v_{x_{j},\sigma}(\lambda, k)}{\abs{k}}\bigg) \quad .
\end{equation}

It is clear from equation \eqref{almostproj} that 
\begin{equation*}
\qm\proj{M}{1}\qm = (\id - \qm)\proj{M}{1}(\id - \qm) = 0,
\end{equation*}
so 
\begin{equation*}
(\proj{(1)}{M})^{2} - \proj{(1)}{M} = \Or(\epsi^{2}) \quad .
\end{equation*}

$\proj{(1)}{M}$ is also almost invariant for the total dynamics, in the sense that
\begin{equation*}\begin{split}
& [\proj{(1)}{M}, \hepsisigma] = [\qm, \hiepsisigma{0}] + \epsi[\proj{M}{1}, \hiepsisigma{0}] + \epsi[\qm, \hiepsisigma{1}] + \epsi^{2}[\proj{M}{1}, \hiepsisigma{1}] +\\
&+ \epsi^{2}[\qm, \hiepsisigma{2}] +\epsi^{3}[\proj{M}{1}, \hiepsisigma{2}]=\Or(\epsi^{2}\sqrt{\log(\sigma^{-1})})\quad.
\end{split}
\end{equation*}

To justify this claim we note that
\begin{equation*}\begin{split}
& [\proj{M}{1}, \hiepsisigma{0}] + [\qm, \hiepsisigma{1}] = \I\sum_{j=1}^{N}\frac{e_{j}}{m_{j}}\big[\pop_{j}\cdot[\qm, \phijsigma], \hiepsisigma{0}\big] - \sum_{j=1}^{N}\frac{e_{j}}{m_{j}}\pop_{j}\cdot\\
&\cdot[\qm, \Phi(v_{x_{j}, \sigma})] = \I\sum_{j,l=1}^{N}\frac{e_{j}}{m_{j}}\big[\pop_{j}\cdot[\qm, \phijsigma], \frac{\pop_{l}^{2}}{2}\big] +\\
&+ \I\sum_{j=1}^{N}\frac{e_{j}}{m_{j}}\big[\pop_{j}\cdot[\qm, \phijsigma], \vcoulsigma\big] + \I\sum_{j=1}^{N}\frac{e_{j}}{m_{j}}\big[\pop_{j}\cdot[\qm, \phijsigma], \Hf\big] +\\
&- \sum_{j=1}^{N}\frac{e_{j}}{m_{j}}\pop_{j}\cdot[\qm, \Phi(v_{x_{j}, \sigma})] = \Or(\epsi\sqrt{\log(\sigma^{-1})}) +\\
&+ \I\sum_{j=1}^{N}\frac{e_{j}}{m_{j}}\pop_{j}\cdot\big[\qm, [\phijsigma, \Hf]\big] - \sum_{j=1}^{N}\frac{e_{j}}{m_{j}}\pop_{j}\cdot[\qm, \Phi(v_{x_{j}, \sigma})] =\\
&= \Or(\epsi\sqrt{\log(\sigma^{-1})}) \quad .
\end{split}
\end{equation*}

To analyze in a simple way the restriction of the dynamics to the subspace defined by $\proj{(1)}{M}$ one builds an almost unitary $U^{(1)}$, which maps the almost projections to a reference projection up to terms of order $\Or(\epsi^{2})$. Using the formal expression we get for $U^{(1)}$, we will define in next section a true unitary operator which will allow us to construct a rigorous version of the almost projections $\proj{(1)}{M}$.

A natural choice for the reference projections, linked to the physics of the system, is to choose them equal to the $\qm$s. We assume then that also $U^{(1)}$ can be expanded in powers of $\epsi$,
\begin{equation*}
U^{(1)}:= \id + \epsi U_{1},  
\end{equation*}
with the condition $U_{1} + U_{1}^{*}=\Or(\epsi)$. This ensures that
\begin{equation*}
U^{(1)}U^{(1)*}= \Or(\epsi^{2}),\qquad U^{(1)*}U^{(1)}=\Or(\epsi^{2})\quad.
\end{equation*}

To determine $U_{1}$ we impose that $U^{(1)}$ intertwines the almost invariant projections with the reference projections $\qm$ up to terms of order $\Or(\epsi^{2})$:
\begin{equation*}
U^{(1)}\proj{(1)}{M}U^{(1)*}\eqex \qm + \Or(\epsi^{2}) \quad .
\end{equation*}

The left-hand side gives
\begin{equation*}\begin{split}
& U^{(1)}\proj{(1)}{M}U^{(1)*}=(\id + \epsi U_{1})(\qm + \epsi\proj{M}{1})(\id - \epsi U_{1})+\Or(\epsi^{2})=\\
&=\qm + \epsi([U_{1}, \qm] + \proj{M}{1}) + \Or(\epsi^{2}) = \qm + \epsi\bigg([U_{1}, \qm] +\\
&- \bigg[\I\sum_{j=1}^{N}\frac{e_{j}}{m_{j}}\phijsigma\cdot\pop_{j}, \qm\bigg]\bigg) + \Or(\epsi^{2}),
\end{split}
\end{equation*}
so we can choose
\begin{equation}
U^{(1)} = \id + \I\epsi\sum_{j=1}^{N}\frac{e_{j}}{m_{j}}\phijsigma\cdot\pop_{j}\quad .
\end{equation}

\subsection{Rigorous definition}\label{unitaryrigorous}

To get a well-defined unitary operator from the formal expression for $U^{(1)}$ we first cutoff the number of photons in the field operator $\phijsigma$, replacing it by
\begin{equation}
\phijsigmal:= \qless{L}\phijsigma\qless{L},
\end{equation} 
where $L$ is fixed, but otherwise arbitrary. 

We introduce then a cutoff in the total energy, to cope with the unboundedness of the momentum of the electrons $\pop_{j}$. This reflects the fact that the adiabatic approximation holds uniformly only on states where the kinetic energy of the slow particles in uniformly bounded.

More precisely, given a function $\chi\in\coinf(\field{R}{})$, we define
\begin{equation}\label{uonelchi}\begin{split}
& \uonelchi := \id + \I\epsi\sum_{j=1}^{N}\frac{e_{j}}{m_{j}}\phijsigmal\cdot\pop_{j} - \I\epsi(\id - \chi(\hepsisigma))\sum_{j=1}^{N}\frac{e_{j}}{m_{j}}\phijsigmal\cdot\pop_{j}(\id - \chi(\hepsisigma))\\
&=\id + \epsi\chi(\hepsisigma)\uonelsigma + \epsi(\id - \chi(\hepsisigma))\uonelsigma\chi(\hepsisigma),
\end{split}
\end{equation}
where we have defined 
\begin{equation}
\uonelsigma:= \I\sum_{j=1}^{N}\frac{e_{j}}{m_{j}}\phijsigmal\cdot\pop_{j}\quad.
\end{equation}

Note that it follows from the Coulomb gauge condition that $\uonelsigmastar = -\uonelsigma$, so
\begin{equation}\label{uonelchistar}
\uonelchistar = \id - \epsi\chi(\hepsisigma)\uonelsigma - \epsi(\id - \chi(\hepsisigma))\uonelsigma\chi(\hepsisigma)\quad.
\end{equation}
\begin{lemma}\label{lemmaphisigma}\begin{enumerate}
\item For each $x\in\field{R}{3}$, $\phijsigmal(x)\in\mathcal{L}(\fock)$, and $\phijsigmal(x)^{*}=\phijsigmal(x)$. 

Moreover,
\begin{equation}
\phijsigmal: \field{R}{3}\to\mathcal{L}(\fock), \qquad x\mapsto \phijsigmal(x), \quad\in C_{\mathrm{b}}^{\infty}(\field{R}{3}, \mathcal{L}(\fock)),
\end{equation}
and, for $\sigma$ small enough,
\begin{equation}
\norm{\phijsigmal}_{\mathcal{L}(\hilbert{H})}\leq C\sqrt{\mathrm{L}+1}\sqrt{\log(\sigma^{-1})}.
\end{equation}

Given $\alpha\in\field{N}{3}$ with $\abs{\alpha}>0$, it holds instead 
\begin{equation}
\deriv{x}{\alpha}\phijsigmal = \deriv{x}{\alpha}\phijzerol + \Or(\sigma^{\abs{\alpha}}\sqrt{\mathrm{L}+1})_{\mathcal{L}(\hilbert{H})},
\end{equation}
where
\begin{equation}
\deriv{x}{\alpha}\phijzerol := (\deriv{x}{\alpha}\phijsigmal)_{|\sigma=0}
\end{equation}
is a well-defined bounded operator on $\hilbert{H}$.

\item The statements of point $1$ (except for the self-adjointness of $\phijsigmal(x)$) remain true if $\fock$ is replaced by $D(\Hf)$.
\end{enumerate} 
\end{lemma}
\begin{corollary}
The fibered operators $\deriv{x}{\alpha}\phijsigmal$ belong to $\mathcal{L}(\hilbert{H})\cap\mathcal{L}(\hilbert{H}_{0})$ $\forall\alpha\in\field{N}{3}$.
\end{corollary}

\begin{proof}
The proof of both statements follows from the facts that
\begin{equation*}
\norm{\qless{L}\Phi(g_{x}(\cdot))\qless{L}}_{\mathcal{L}(\hilbert{H})}\leq 2^{1/2}\sqrt{L+1}\sup_{x\in\field{R}{3}}\norm{g_{x}(\cdot)}_{L^{2}(\field{R}{3}\times\ztwo)}
\end{equation*}
and that
\begin{equation*}
\deriv{x}{\alpha}v_{x}(k, \lambda)\abs{k}^{-1} = (-\I)^{\abs{\alpha}}\abs{k}^{-3/2 + \abs{\alpha}}e_{\lambda}(k)\hat{\varphi}(k)\E^{-\I k\cdot x} \quad .
\end{equation*}
\end{proof}

\begin{lemma}
The operator $\uonelchi$ is closable and its closure, which we denote by the same symbol, belongs to $\mathcal{L}(\hilbert{H})\cap\mathcal{L}(\hilbert{H}_{0})$. Moreover
\begin{equation}
\norm{\uonelchi}_{\mathcal{L}(\hilbert{K})}\leq C(1+\epsi\sqrt{\log(\sigma^{-1})}),
\end{equation}
where $\hilbert{K}=\hilbert{H}$ or $\hilbert{H}_{0}$. The same holds for $\uonelchistar$.
\end{lemma}

\begin{proof}
The operator
\begin{equation*}
\chi(\hepsisigma)\phijsigmal\cdot\pop_{j}
\end{equation*}
is defined on $D(\pop_{j})$ and, since $\chi(\hepsisigma)\phijsigmal$ is a bounded operator, we have
\begin{equation*}
(\chi(\hepsisigma)\phijsigmal\cdot\pop_{j})^{*}=\pop_{j}\cdot\phijsigmal\chi(\hepsisigma)
\end{equation*}
which is clearly bounded. This shows that $\chi(\hepsisigma)\phijsigmal\cdot\pop_{j}$ is closable and its closure belongs to $\mathcal{L}(\hilbert{H})$.

The same reasoning can be applied to the operator
\begin{equation*}
\hepsisigma\phijsigmal\cdot\pop_{j}\chi(\hepsisigma)
\end{equation*}
which shows that $\uonelchi$ is also in $\mathcal{L}(\hilbert{H}_{0})$.

The estimate on the norm follows now from the estimate on the norm of $\phijsigmal$ given in lemma \ref{lemmaphisigma}.
\end{proof}

\begin{theorem}\label{unitaryintertwiner}
Assume that $\sigma=\sigma(\epsi)$ satisfies the condition
\begin{equation}
\epsi\sqrt{\log(\sigma(\epsi)^{-1})}\to 0, \qquad\epsi\to 0^{+},
\end{equation}
then the operator
\begin{equation}
\hilbert{U}:= \uonelchi[\uonelchistar\uonelchi]^{-1/2}
\end{equation}
is well-defined and unitary, for $\epsi$ small enough.

Both $\hilbert{U}$ and $\hilbert{U}^{*}$ belong to $\mathcal{L}(\hilbert{H})\cap\mathcal{L}(\hilbert{H}_{0})$, with the property that
\begin{equation}\label{normubounded}
\norm{\hilbert{U}}_{\mathcal{L}(\hilbert{H}_{0})},\qquad \norm{\hilbert{U}^{*}}_{\mathcal{L}(\hilbert{H}_{0})}\quad \leq C,
\end{equation}
where $C$ is independent of $\epsi$ and $\sigma$.

Moreover we can expand them in powers of $\epsi$ and the corresponding series converges both in $\mathcal{L}(\hilbert{H})$ and $\mathcal{L}(\hilbert{H}_{0})$.
\end{theorem}
\begin{proof}
It follows from equations \eqref{uonelchi} and \eqref{uonelchistar} that, defining
\begin{equation}\label{tsigma}
T_{\sigma}:= \chi(\hepsisigma)\uonelsigma + (\id - \chi(\hepsisigma))\uonelsigma\chi(\hepsisigma),
\end{equation}
we have then
\begin{equation*}
\uonelchi = \id + \epsi T_{\sigma(\epsi)}, \qquad T_{\sigma}^{*}=-T_{\sigma},\, \norm{T_{\sigma}}_{\mathcal{L}(\hilbert{K})}\leq C\sqrt{\log(\sigma^{-1})},
\end{equation*}
where $\hilbert{K}=\hilbert{H}$ or $\hilbert{H}_{0}$.

From this expression we get immediately that
\begin{equation*}
\uonelchistar\uonelchi = \id - \epsi^{2}T_{\sigma(\epsi)}^{2},
\end{equation*}
so, choosing $\epsi$ small enough, we have that $\epsi^{2}\norm{T_{\sigma(\epsi)}^{2}}_{\mathcal{L}(\hilbert{K})}<1$, therefore the square root is well-defined, and can be expressed through a convergent power series in $\mathcal{L}(\hilbert{K})$:
\begin{equation}\label{squareroot}
(\id - \epsi^{2}T_{\sigma(\epsi)}^{2})^{-1/2}=\sum_{j=0}^{\infty}\frac{(2j-1)!!}{(2j)!!}\epsi^{2j}T_{\sigma(\epsi)}^{2j}\quad.
\end{equation}

From standard calculations it follows in the end that $\hilbert{U}$ is unitary on $\hilbert{H}$.
\end{proof}

\section{The dressed Hamiltonian}\label{dressed}

We define the dressed Hamiltonian as the unitary transform of $\hepsisigma$,
\begin{equation}\label{hdres}
\hdres:= \hilbert{U}\hepsisigma\hilbert{U}^{*} \quad.
\end{equation}

Since $\hilbert{U}$ is a bijection on $\hilbert{H}_{0}$, $\hdres$ is self-adjoint on $\hilbert{H}_{0}$, and using the expansion of $\hilbert{U}$ on $\mathcal{L}(\hilbert{H}_{0})$, we can expand $\hdres$ in $\mathcal{L}(\hilbert{H}_{0}, \hilbert{H})$. 

\begin{theorem}\label{thmhdres}
Assume that $\sigma = \sigma(\epsi)$ satisfies conditions \eqref{conditionssigma}. The expansion of the dressed Hamiltonian up to the second order is then given by
\begin{equation}
\hdres = \hiepsidres{0} + \epsi\hiepsidres{1} + \epsi^{2}\hiepsidres{2} + \Or\big(\epsi^{3}(\log\sigma^{-1})^{3/2}\big)_{\mathcal{L}(\hilbert{H}_{0}, \hilbert{H})},
\end{equation}
where
\begin{equation*}
\hiepsidres{0} = \hiepsi{0} = \sum_{j=1}^{N} \frac{1}{2m_{j}}\pop_{j}^{2} + \vcoul + \Hf,
\end{equation*}
$\hiepsidres{1}$ is given in equation \eqref{honedres} and $\hiepsidres{2}$ in equation \eqref{htwodres}.
\end{theorem}
\begin{proof}
Applying equation \eqref{squareroot}, we get that
\begin{equation*}
\hilbert{U}= \id +\epsi T +\frac{\epsi^{2}}{2}T^{2} + \Or(\epsi^{3}T^{3})\footnote[1]{The expansion of $\hilbert{U}$ till the second order coincides with that of $\E^{\epsi T}$, however $\E^{\epsi T}$ is \emph{not} the correct dressing transformation to every order. Following the formal procedure sketched in section \ref{unitaryformal} one can construct a second order expression for $\hilbert{U}$, which depends however also in general on the second order coefficient of the Hamiltonian, $\hiepsisigma{2}$ and has not a simple exponential form.},
\end{equation*}
therefore 
\begin{equation*}\begin{split}
& \hdres = (\id +\epsi T +\epsi^{2}T^{2}/2)\hepsisigma(\id -\epsi T +\epsi^{2}T^{2}/2) +\\ &+\Or\big(\epsi^{3}(\log\sigma^{-1})^{3/2}\big)_{\mathcal{L}(\hilbert{H}_{0}, \hilbert{H})}=\\
&=\hepsisigma +\epsi[T, \hepsisigma] +\frac{\epsi^{2}}{2}[T, [T, \hepsisigma]] +\Or\big(\epsi^{3}(\log\sigma^{-1})^{3/2}\big)_{\mathcal{L}(\hilbert{H}_{0}, \hilbert{H})}\stackrel{\eqref{tsigma}}{=}\\
&=\hiepsisigma{0} + \epsi\big(\hiepsisigma{1} + \chi(\hepsisigma)[\uonelsigma, \hiepsisigma{0}] + ``(\id-\chi)\cdots\chi\textrm{''}\big) +\\
&+ \epsi^{2}\bigg(\hiepsisigma{2} + \chi(\hepsisigma)[\uonelsigma, \hiepsisigma{1}] +\frac{1}{2}\big[T, \chi(\hepsisigma)[\uonelsigma, \hiepsisigma{0}]\big] +\\
&+ ``(\id-\chi)\cdots\chi\textrm{''}\bigg) +\Or\big(\epsi^{3}(\log\sigma^{-1})^{3/2}\big)_{\mathcal{L}(\hilbert{H}_{0}, \hilbert{H})}=:\\
&=: \hiepsidres{0} + \epsi\hiepsidres{1} + \epsi^{2}\hiepsidres{2} + \Or\big(\epsi^{3}(\log\sigma^{-1})^{3/2}\big)_{\mathcal{L}(\hilbert{H}_{0}, \hilbert{H})}
\end{split}
\end{equation*}
where $``(\id-\chi)\cdots\chi\textrm{''}$ indicates that for every term containing $\chi\cdots$ we have to add a corresponding term containing $(\id-\chi)\cdots\chi$, as in equation \eqref{tsigma}.

Using equation \eqref{estimatecoulomb} to eliminate the $\sigma$, we get immediately
\begin{equation}
\hiepsidres{0} = \hiepsi{0} \quad .
\end{equation}

The commutator in the term of order $\epsi$ gives:
\begin{equation}\label{commuonehzero}\begin{split}
& [\uonelsigma, \hiepsisigma{0}] = -\qless{L}\hiepsisigma{1}\qless{L} + \epsi\sum_{j=1}^{N}\frac{e_{j}}{m_{j}}\phijsigmal\cdot\nabla_{x_{j}}\vcoulsigma +\\ &+ \epsi\sum_{j,l=1}^{N}\frac{e_{j}}{m_{l}m_{j}}\nabla_{x_{l}}(\phijsigmal\cdot\pop_{j})\cdot\pop_{l} + \Or(\epsi^{2}),
\end{split}
\end{equation}
therefore, taking into account equation \eqref{estimateepsisigma},
\begin{equation}\label{honedres}\begin{split}
& \hiepsidres{1} = \big(\id-\chi(\hepsisigma)\big)\qless{L}\hiepsi{1}\qless{L}\big(\id-\chi(\hepsisigma)\big) + \qless{L}\hiepsi{1}\qmore{L} + \qmore{L}\hiepsi{1}\qless{L} +\\
&+ \qmore{L}\hiepsi{1}\qmore{L}\quad.
\end{split}
\end{equation}
We will show below that this term vanishes in the effective dynamics.

Concerning the terms of second order we have
\begin{equation}\label{commuuonehone}\begin{split}
& [\uonelsigma, \hiepsisigma{1}] = -\I\sum_{j,l=1}^{N}\frac{e_{j}e_{l}}{m_{j}m_{l}}[\phijsigmal\cdot\pop_{j}, \Phi(v_{x_{l}, \sigma})\cdot\pop_{l}] =\\ &=-\I\sum_{j,l=1}^{N}\frac{e_{j}e_{l}}{m_{j}m_{l}}\qless{L}[\phijsigma^{\alpha}, \Phi(v_{x_{l}, \sigma}^{\beta})]\qless{L}\pop_{j}^{\alpha}\pop_{j}^{\beta} + \mathcal{R}_{\mathrm{L-1}} + \Or(\epsi)_{\mathcal{L}(\hilbert{H}_{0}, \hilbert{H})},
\end{split}
\end{equation}
where $\mathcal{R}_{\mathrm{L-1}}$ is a term which vanishes on the range of $Q_{j}$ when $j<\mathrm{L}-1$.

Using equation \eqref{commphi} to calculate the commutator of the two field operator we get in the end
\begin{equation}\label{darwin}\begin{split}
& -\I\sum_{j,l=1}^{N}\frac{e_{j}e_{l}}{m_{j}m_{l}}[\phijsigma^{\alpha}, \Phi(v_{x_{l}, \sigma}^{\beta})]\pop_{j}^{\alpha}\pop_{j}^{\beta} = \I\sum_{j,l=1}^{N}\frac{e_{j}e_{l}}{m_{j}m_{l}}\Re \langle \frac{v_{x_{j}, \sigma}^{\alpha}}{\abs{k}}, v_{x_{l}, \sigma}^{\beta}\rangle_{L^{2}(\field{R}{3}, dk)\otimes\field{C}{2}}\cdot\\
&\cdot\pop_{j}^{\alpha}\pop_{j}^{\beta}=\\
&=\I\sum_{j,l=1}^{N}\frac{e_{j}e_{l}}{m_{j}m_{l}}\int_{\field{R}{3}}dk\, \frac{\abs{\hat{\varphi}(k)}^{2}}{\abs{k}^{2}}\E^{\I k\cdot(x_{j}-x_{l})}\pop_{j}\cdot(\id - \kappa\otimes\kappa)\pop_{l} =\\
&= \I\sum_{j,l=1}^{N}\frac{e_{j}e_{l}}{m_{j}m_{l}}\int_{\field{R}{3}}dk\, \frac{\abs{\hat{\varphi}(k)}^{2}}{\abs{k}^{2}}\E^{\I k\cdot x_{j}}\pop_{j}\cdot(\id - \kappa\otimes\kappa)\pop_{l}\E^{-\I k\cdot x_{l}}
\end{split}
\end{equation}

The remaining term of second order gives,
\begin{equation*}\begin{split}
& \big[T, \chi(\hepsisigma)[\uonelsigma, \hiepsisigma{0}]\big] \stackrel{\eqref{commuonehzero}}{=} -\big[T, \chi(\hepsisigma)\qless{L}\hiepsisigma{1}\qless{L}\big] + \Or(\epsi)\stackrel{\eqref{tsigma}}{=}\\ &=-\chi[\uonelsigma, \chi]\qless{L}\hiepsisigma{1}\qless{L} - \chi^{2}\qless{L}[\uonelsigma, \hiepsisigma{1}]\qless{L} - \chi[\chi, \qless{L}\hiepsisigma{1}\qless{L}]\uonelsigma 
\end{split}
\end{equation*}

Putting the calculations above together, and using equations \eqref{infraredspin} and \eqref{infraredasquare}, we get in the end
\begin{equation}\label{htwodres}\begin{split}
& \hiepsidres{2} = \hiepsi{2} + \sum_{j=1}^{N}\frac{e_{j}}{m_{j}}\phijsigmal\cdot\nabla_{x_{j}}\vcoul + \sum_{j, l =1}^{N}\frac{e_{j}}{m_{l}m_{j}}\nabla_{x_{l}}(\phijsigmal\cdot\pop_{j})\cdot\pop_{l} + \\
&+ \chi[\uonelsigma, \hiepsisigma{1}] -\frac{1}{2}\chi^{2}\qless{L}[\uonelsigma, \hiepsisigma{1}]\qless{L} + (\id - \chi)[\uonelsigma, \hiepsisigma{1}]\chi +\\
&- \frac{1}{2}\chi[\chi, \qless{L}\hiepsisigma{1}\qless{L}]\uonelsigma -\frac{1}{2}\chi[\uonelsigma, \chi]\qless{L}\hiepsisigma{1}\qless{L},
\end{split}
\end{equation}
where the commutator $[\uonelsigma, \hiepsisigma{1}]$ is given in equations \eqref{commuuonehone} and \eqref{darwin}. We will show below that only few of the terms written above contribute to the effective dynamics, giving the correction to the mass of the electrons and the Darwin term.
\end{proof}

\section{The effective dynamics}\label{effective}

We quote without proof a number of lemmas, which, with minor modifications, are identical to the ones proved in \cite{TeTe}.

\begin{lemma}\label{lemmaxi}(see corollary $4$ \cite{TeTe})

Given a function $\chi\in\coinf(\field{R}{})$ and a $\sigma>0$, we have
\begin{equation}\label{energynotincreas}
a\bigg(\frac{\I v_{x_{j}, \sigma}(\lambda, k)}{\abs{k}}\bigg)\qm\chi(\hiepsi{0})=\mathrm{Q}_{M-1}\xi(\hiepsi{0})a\bigg(\frac{\I v_{x_{j}, \sigma}(x, \lambda)}{\abs{k}}\bigg)\qm\chi(\hiepsi{0}) + \Or_{0}(\epsi^{\infty}),
\end{equation}
where $\xi\in\coinf(\field{R}{})$, $\xi\chi=\chi$ and
\begin{equation*}
c_{\xi}=2d_{\chi} + E_{\infty},
\end{equation*}
where 
\begin{eqnarray*}
d_{\chi}&:=& 2c_{\chi}+E_{\infty}+\min\{c_{\chi}, \Lambda\}\,,\\
 c_{\chi}&:=&\sup\{\abs{k}: k\in\mathrm{supp}\ \chi\}\,,\\
 E_{\infty}&:=& \sup_{x\in\field{R}{3N}}\abs{\vcoul(x)}\,,
\end{eqnarray*}
and we can choose $\sup\{\abs{k}: k\in\mathrm{supp}\ \xi\}$ arbitrarily close to $c_{\xi}$.

An analogous statement holds for the creation operator.
\end{lemma}

\begin{lemma}\label{approxchi}

Assume that $\sigma$ satisfies conditions \eqref{conditionssigma}, then

$1.$\ Given a function $\tilde{\chi}\in\coinf(\field{R}{})$, we have
\begin{equation}
\tilde{\chi}(\hepsisigma)-\tilde{\chi}(\hiepsi{0})=\epsi\mathcal{R}^{\epsi}_{\chi},
\end{equation}
where $\mathcal{R}^{\epsi}_{\chi}\in\mathcal{L}(\hilbert{H}, \hilbert{H}_{0})$, $\norm{\mathcal{R}^{\epsi}_{\chi}}_{\mathcal{L}(\hilbert{H}, \hilbert{H}_{0})}=\Or(1)$ and
\begin{equation}
\mathcal{R}^{\epsi}_{\chi}\qm\tilde{\chi}(\hiepsi{0}) = (\mathrm{Q}_{M+1}+\mathrm{Q}_{M-1})\xi(\hiepsi{0})\mathcal{R}^{\epsi}_{\chi}\qm\tilde{\chi}(\hiepsi{0}) + \Or(\epsi)_{\mathcal{L}(\hilbert{H}, \hilbert{H}_{0})},
\end{equation}
where $\xi$ has the properties described in lemma \ref{lemmaxi}.

$2.$\ Moreover, we have that
\begin{equation}\label{diffheffhzero}
\tilde{\chi}(\hdres)-\tilde{\chi}(\hiepsi{0})=\Or(\epsi)_{\mathcal{L}(\hilbert{H}, \hilbert{H}_{0})},
\end{equation}
and that
\begin{equation}\label{diffqnheffhzero}
\qm\tilde{\chi}(\hdres)=\qm\tilde{\chi}(\hiepsi{0})\tilde{\tilde{\chi}}(\hdres) +\Or(\epsi^{2}\sqrt{\log(\sigma^{-1})})_{\mathcal{L}(\hilbert{H}, \hilbert{H}_{0})},
\end{equation}
where $\tilde{\tilde{\chi}}$ is any $\coinf(\field{R}{})$ function such that $\tilde{\chi}\tilde{\tilde{\chi}}=\tilde{\chi}$ and $\tilde{\tilde{\chi}}\chi=\tilde{\tilde{\chi}}$, $M<\mathrm{L}-1$.
\end{lemma}
\begin{proof}
We give the proof just for point $1$. The proof of point $2$ is analogous and can be found also in (\cite{TeTe}, lemma $7$).

Proceeding as in lemma \ref{lemmachiinfra}, we apply Hellfer-Sj\"ostrand formula and we get
\begin{equation*}
\chi(\hepsisigma) - \chi(\hiepsi{0}) = \frac{1}{\pi}\int_{\field{R}{2}}dxdy\, \bar{\partial}\chi^{a}(z)\big[(\hepsisigma-z)^{-1} - (\hiepsi{0}-z)^{-1}\big].
\end{equation*}

Since both Hamiltonians are self-adjoint on the same domain we get, iterating equation \eqref{differenceresolvents},
\begin{equation*}
(\hepsisigma-z)^{-1} - (\hiepsi{0}-z)^{-1} = -\epsi(\hiepsi{0}-z)^{-1}\hiepsisigma{1}(\hiepsi{0}-z)^{-1} + \Or(\epsi^{2}\abs{\Im z}^{-3})_{\mathcal{L}(\hilbert{H}, \hilbert{H}_{0})}.
\end{equation*}

The statement now follows from the explicit expression of $\hiepsisigma{1}$ using lemma \ref{lemmaxi}.
\end{proof}

\begin{theorem}\emph{(Zero order approximation to the time evolution)} 
\begin{equation}
\norm{(\E^{-\I\hdres \frac{t}{\epsi}} - \E^{-\I\hiepsi{0}\frac{t}{\epsi}})\qm\tilde{\chi}(\hdres)}_{\mathcal{L}(\hilbert{H})} = \Or(\sqrt{M+1}\abs{t}\epsi\sqrt{\log(\sigma(\epsi)^{-1})}),
\end{equation}
\begin{equation}
\norm{\qm(\E^{-\I\hdres \frac{t}{\epsi}} - \E^{-\I\hiepsi{0}\frac{t}{\epsi}})\tilde{\chi}(\hdres)}_{\mathcal{L}(\hilbert{H})} = \Or(\sqrt{M+1}\abs{t}\epsi\sqrt{\log(\sigma(\epsi)^{-1})}),
\end{equation}
for every $\tilde{\chi}\in\coinf(\field{R}{})$ such that $\tilde{\chi}\chi=\tilde{\chi}$.
\end{theorem}

\begin{proof}
Using lemma \ref{approxchi} we replace $\chi(\hdres)$ with $\chi(\hiepsi{0})$, since the difference, being of order $\Or(\epsi)$, is smaller than the error we want to prove.

Both Hamiltonians are self-adjoint on $\hilbert{H}_{0}$, therefore, applying Duhamel formula, we get:
\begin{eqnarray*} 
\lefteqn{  (\E^{-\I t\hdres/{\epsi}}-\E^{-\I t \hiepsi{0}/{\epsi}})\qm\tilde{\chi}(\hiepsi{0}) =}\\&=& -\frac{\I}{\epsi} \int_{0}^{t} ds\, 
\E^{\I (s-t)\hdres/\epsi}
(\hdres-\hiepsi{0})\E^{-\I s \hiepsi{0}/\epsi}\qm\tilde{\chi}(\hiepsi{0})=\\&=& -\I\int_{0}^{t}ds\, \E^{\I (s-t)\hdres/\epsi}h_{1,\chi}\E^{-\I s h_{0}/\epsi}\qm\tilde{\chi}(\hiepsi{0}) 
+ \Or(\epsi\sqrt{M+1}\sqrt{\log(\sigma^{-1})})_{\mathcal{L}(\hilbert{H})}\\
&=&-\I\int_{0}^{t}ds\, \E^{\I (s-t)\hdres/\epsi}\hiepsidres{1} \qm\tilde{\chi}(\hiepsi{0})\E^{-\I s h_{0}/\epsi} + \Or(\epsi\sqrt{M+1}\sqrt{\log(\sigma^{-1})})_{\mathcal{L}(\hilbert{H})}.
\end{eqnarray*}
Putting in the previous equation the expression of $\hiepsidres{1}$, equation \eqref{honedres}, and observing that we can replace, again by lemma \ref{approxchi}, $\chi(\hepsisigma)$ with $\chi(\hiepsi{0})$, we get that the right-hand side is of order $\Or(\sqrt{M+1}\abs{t}\epsi\sqrt{\log(\sigma^{-1})})$.

For the second estimate we apply again Duhamel formula, inverting the position of the unitaries,
\begin{eqnarray*}
\lefteqn{ \qm(\E^{-\I t\hdres/\epsi}-\E^{-\I t\hiepsi{0}/\epsi})\tilde{\chi}(\hdres) =}\\&=& -\frac{\I}{\epsi}\int_{0}^{t}ds\, \qm\E^{\I (s-t)\hiepsi{0}/\epsi}(\hdres-\hiepsi{0}) \E^{-\I s \hdres/\epsi}\tilde{\chi}(\hdres)=\\&=&-\I\int_{0}^{t}ds\, \E^{\I (s-t)\hiepsi{0}/\epsi}\qm\hiepsidres{1}\tilde{\chi}(\hdres)\E^{-\I s \hdres/\epsi} 
+\Or(\epsi\sqrt{M+1}\sqrt{\log(\sigma^{-1})})\,,
\end{eqnarray*}
so, replacing $\chi(\hdres)$ and $\chi(\hepsisigma)$ by $\chi(\hiepsi{0})$, our claim is proved.
\end{proof}

\begin{corollary}\label{adibaticinvariancepepsim}
The dressed projectors $\pepsim$ are almost invariant with respect to the original dynamics, i. e.,
\begin{equation*}
\norm{[\E^{-\I\hepsi\frac{t}{\epsi}}, \pepsim]\chi(\hepsi)}_{\mathcal{L}(\hilbert{H})} = \Or\big(\sqrt{M+1}\abs{t}\epsi\sqrt{\log(\sigma(\epsi)^{-1})}\big)\quad .
\end{equation*}
\end{corollary}
\begin{proof}
From the previous theorem it follows that 
\begin{equation*}
\norm{[\E^{-\I\hdres\frac{t}{\epsi}}, \qm]\chi(\hdres)}_{\mathcal{L}(\hilbert{H})} = \Or\big(\sqrt{M+1}\abs{t}\epsi\sqrt{\log(\sigma(\epsi)^{-1})}\big)\quad .
\end{equation*}

From the definition of $\hdres$, equation \eqref{hdres}, we deduce therefore that
\begin{equation*}
\norm{[\E^{-\I\hepsisigma\frac{t}{\epsi}}, \pepsim]\chi(\hepsisigma)}_{\mathcal{L}(\hilbert{H})} = \Or\big(\sqrt{M+1}\abs{t}\epsi\sqrt{\log(\sigma(\epsi)^{-1})}\big)\quad ,
\end{equation*}
but
\begin{equation*}\begin{split}
& [\E^{-\I\hepsi\frac{t}{\epsi}}, \pepsim]\chi(\hepsi) = [(\E^{-\I\hepsi\frac{t}{\epsi}} - \E^{-\I\hepsisigma\frac{t}{\epsi}}), \pepsim]\chi(\hepsi) + \\
& + [\E^{-\I\hepsisigma\frac{t}{\epsi}}, \pepsim](\chi(\hepsi) - \chi(\hepsisigma)) + [\E^{-\I\hepsisigma\frac{t}{\epsi}}, \pepsim]\chi(\hepsisigma)=\\
& = \Or(\sigma(\epsi)^{1/2}\abs{t})_{\mathcal{L}(\hilbert{H})} + \Or(\epsi\sigma(\epsi)^{1/2})_{\mathcal{L}(\hilbert{H})} +\\ &+ \Or\big(\sqrt{M+1}\abs{t}\epsi\sqrt{\log(\sigma(\epsi)^{-1})}\big)_{\mathcal{L}(\hilbert{H})},
\end{split}
\end{equation*}
where we have used equations \eqref{cutofftimeev} and \eqref{cutoffchi}, the fact that $\chi(\hepsi)\in\mathcal{L}(\hilbert{H}, \hilbert{H}_{0})$ with norm uniformly bounded in $\epsi$ and equation \eqref{normubounded}.
\end{proof}

\begin{lemma}\label{lemtrunchdres}
The truncated dressed Hamiltonian
\begin{equation}
\hdrestwo:= \hiepsi{0} + \epsi\hiepsidres{1} + \epsi^{2}\hiepsidres{2}
\end{equation}
is self-adjoint on $\hilbert{H}_{0}$ for $\epsi$ small enough.
\end{lemma}
\begin{proof}
The proof follows from a symmetric version of the Kato theorem (\cite{ReSi2}, theorem X.13).
\end{proof}

\begin{theorem}\label{firstordertimeevolution}\emph{(First order approximation to the time evolution)} Given a function $\tilde{\chi}\in\coinf(\field{R}{})$,
\begin{equation}\label{firstordertime}\begin{split} 
& \E^{-\I t\hdres/\epsi}\qm\tilde{\chi}(\hdres)= \\
& = \E^{-\I t\hdiag/\epsi}\qm\tilde{\chi}(\hdres)-\I\epsi\int_{0}^{t}ds\, \E^{\I(s-t)\hiepsi{0}/\epsi}h_{2, \mathrm{OD}}\E^{-\I s\hiepsi{0}/\epsi}\qm\tilde{\chi}(\hdres) \\
&+ \Or(\epsi^{3/2}\abs{t})_{\mathcal{L}(\hilbert{H})}(1-\delta_{M0}) +\,\Or\big(\epsi^{2}\sqrt{\log(\sigma^{-1})}(\abs{t} + \abs{t}^{2}))_{\mathcal{L}(\hilbert{H})} \,,
\end{split}
\end{equation}
where $\delta_{M0}=1$ when $M=0$, $0$ otherwise, $[\hdiag, \qm]=0\,\,\, \forall\, M$,
\begin{equation}\label{hdiag}\begin{split}
\hdiag &:= \sum_{j=1}^{N} \frac{1}{2m_{j}}\pop_{j}^{2} + \vcoul + \Hf +\\
&-\epsi^{2}\sum_{l, j=1}^{N}\frac{e_{j}e_{l}}{m_{j}m_{l}}\int_{\field{R}{3}}dk\, \frac{\abs{\hat{\varphi}(k)}^{2}}{2\abs{k}^{2}}\E^{\I k\cdot x_{j}}\pop_{j}\cdot(\id - \kappa\otimes\kappa)\pop_{l}\E^{-\I k\cdot x_{l}} =\\
&= \hiepsi{0} + \epsi^{2}\vdarw \quad.
\end{split}
\end{equation}
The off-diagonal Hamiltonian is defined by
\begin{equation}
h_{2, \mathrm{OD}}:= \sum_{j=1}^{N}\frac{e_{j}}{m_{j}}\phijsigma\cdot\nabla_{x_{j}}\vcoul\, .
\end{equation}
\end{theorem}

\begin{remark}\label{remarkspin}
The spin term, as mentioned in the introduction, does not appear in the effective dynamics, even though it is apparently of order $\Or(\epsi^{2})$. This is due to the fact that it is off diagonal with respect to the decomposition of the Hilbert space associated to the $\qm$s and that the coupling function in the magnetic field, equation \eqref{magneticfield}, goes to zero like $\abs{k}^{1/2}$ when $\abs{k}\to 0^{+}$. This implies that the term is actually smaller than $\Or(\epsi^{2})$, as explained in the proof.
\end{remark}

\begin{proof}
We split the proof into three parts. In the first one, we show that equation \eqref{firstordertime} is true with a diagonal Hamiltonian $\hdiagtilde$ given by
\begin{equation}\label{hdiagtilde}\begin{split}
\hdiagtilde &:= \sum_{j=1}^{N} \frac{1}{2m_{j}}\pop_{j}^{2} + \vcoul + \Hf + \epsi^{2}\sum_{j=1}^{N}\frac{e_{j}^{2}}{2m_{j}}a(v_{x_{j}})^{*}a(v_{x_{j}}) +\\
&-\epsi^{2}\sum_{l, j=1}^{N}\frac{e_{j}e_{l}}{m_{j}m_{l}}\int_{\field{R}{3}}dk\, \frac{\abs{\hat{\varphi}(k)}^{2}}{2\abs{k}^{2}}\E^{\I k\cdot x_{j}}\pop_{j}\cdot(\id - \kappa\otimes\kappa)\pop_{l}\E^{-\I k\cdot x_{l}} \,,
\end{split}
\end{equation}
and an off-diagonal one $\tilde{h}_{2, \mathrm{OD}}$ defined by
\begin{equation}\label{htwotildeod}\begin{split} 
\tilde{h}_{2, \mathrm{OD}} &:= \sum_{j=1}^{N}\frac{e_{j}}{m_{j}}\phijsigma\cdot\nabla_{x_{j}}\vcoul +  \sum_{j, l =1}^{N} \frac{e_{j}}{m_{l}m_{j}}\nabla_{x_{l}}(\phijsigma\cdot\pop_{j})\cdot\pop_{l} +\\ &+ \sum_{j=1}^{N}\frac{e_{j}}{2m_{j}}\sigma_{j}\cdot\Phi(\I k \times v_{x_{j}}) + \sum_{j=1}^{N}\frac{e_{j}^{2}}{2m_{j}}\big[a(v_{x_{j}})^{2} + a(v_{x_{j}})^{*\,2}] \, .
\end{split}
\end{equation}
In the second part we prove that if one neglects the term
\begin{equation*}
\epsi^{2}\sum_{j=1}^{N}\frac{e_{j}^{2}}{2m_{j}}a(v_{x_{j}})^{*}a(v_{x_{j}})
\end{equation*}
in $\hdiagtilde$, one gets an error of order $\Or(\epsi^{3/2}\abs{t})$ in the time evolution. Note that this term is exactly zero if the initial state for the field is the Fock vacuum.

In the third part, we prove analogously that we can replace $\tilde{h}_{2, \mathrm{OD}}$ with $h_{2, \mathrm{OD}}$. 

More specifically, the terms which we neglect in $\hdiagtilde$ and $\tilde{h}_{2, \mathrm{OD}}$ give rise to higher order contributions to the time evolution, although their norm in $\mathcal{L}(\hilbert{H}_{0}, \hilbert{H})$ is not small. 

This is caused by the fact that they are strongly oscillating in $\abs{k}$, so that their behavior is determined by the value of the density of states in a neighborhood of $k=0$. For all these terms, the density however vanishes for $k=0$, uniformly in $\sigma$, and this implies that they are of lower order with respect to the leading pieces whose density is constant (for the terms in $\hdiag$) or diverges logarithmically in $\sigma$ (for the terms in $h_{2, \mathrm{OD}}$). We elaborate on this last observation in a corollary to this theorem. 

We start showing that we can, up to the desired error, replace $\hdres$ by $\hdrestwo$. 
By lemma \ref{lemtrunchdres}, $\hdrestwo$ is self-adjoint on $\hilbert{H}_{0}$ like $\hdres$, therefore we can apply the Duhamel formula and use theorem \ref{thmhdres} to get
\begin{eqnarray*}
\E^{-\I t\hdres/\epsi}-\E^{-\I t\hdrestwo/\epsi}& = &-\frac{\I}{\epsi}\int_{0}^{t}ds\, \E^{\I (s-t)\hdres/\epsi}(\hdres-\hdrestwo)\E^{-\I s \hdrestwo/\epsi}\\
&=&\Or(\epsi^{2}\big(\log(\sigma^{-1})\big)^{3/2})\,.
\end{eqnarray*}
Moreover, using lemma \ref{approxchi}, we can replace $\qm\tilde{\chi}(\hdres)$ by $\qm\tilde{\chi}(\hiepsi{0})\tilde{\tilde{\chi}}(\hdres)$.

Since the diagonal Hamiltonian $\hdiagtilde$ is also self-adjoint on $\hilbert{H}_{0}$ for $\epsi$ sufficiently small (the proof can be given along the same lines of lemma \ref{lemtrunchdres}), we apply again Duhamel formula, 
\begin{eqnarray*}
\lefteqn{ (\E^{-\I t\hdrestwo/\epsi}-\E^{-\I t\hdiagtilde/\epsi})\qm\tilde{\chi}(\hiepsi{0})\tilde{\tilde{\chi}}(\hdres) =}\\ &=&-\,\frac{\I}{\epsi}\int_{0}^{t}ds\ \E^{\I(s-t)\hdrestwo/\epsi}(\hdrestwo-\hdiagtilde)\E^{-\I s\hdiagtilde/\epsi}\qm\tilde{\chi}(\hiepsi{0})\tilde{\tilde{\chi}}(\hdres)\\
&=&-\,\I\int_{0}^{t}ds\ \E^{\I(s-t)\hdrestwo/\epsi}\hiepsidres{1}\E^{-\I s\hdiagtilde/\epsi}\qm\tilde{\chi}(\hiepsi{0})\tilde{\tilde{\chi}}(\hdres)\\
&&-\,\I\epsi\int_{0}^{t}ds\ \E^{\I(s-t)\hdrestwo/\epsi}(\hiepsidres{2}-h_{2, \mathrm{D}})\E^{-\I s\hdiagtilde/\epsi}\qm\tilde{\chi}(\hiepsi{0})\tilde{\tilde{\chi}}(\hdres)\,.
\end{eqnarray*}
To analyze the first term, we remark that, proceeding as in lemma \ref{approxchi}, one can prove that
\begin{equation*}
\tilde{\chi}(\hiepsi{0})-\tilde{\chi}(\hdiagtilde)= \Or(\epsi^{2}\sqrt{\log(\sigma^{-1})})_{\mathcal{L}(\hilbert{H})},
\end{equation*}
so 
\begin{equation*}
[\E^{-\I t\hdiagtilde/\epsi}, \tilde{\chi}(\hiepsi{0})] = \Or(\epsi^{2}\sqrt{\log(\sigma^{-1})})_{\mathcal{L}(\hilbert{H})},
\end{equation*}
therefore,
\begin{eqnarray*}
\lefteqn{\hspace{-1cm}-\I\int_{0}^{t}ds\ \E^{\I(s-t)\hdrestwo/\epsi}\hiepsidres{1}\E^{-\I s\hdiagtilde/\epsi}\qm\tilde{\chi}(\hiepsi{0})\tilde{\tilde{\chi}}(\hdres)=}\\
&=&-\,\I\int_{0}^{t}ds\ \E^{\I(s-t)\hdrestwo/\epsi}\hiepsidres{1}\qm\tilde{\chi}(\hiepsi{0})\E^{-\I s\hdiagtilde/\epsi}\tilde{\tilde{\chi}}(\hdres) \\
&&+\, \Or(\epsi^{2}\abs{t}\sqrt{\log(\sigma^{-1})})_{\mathcal{L}(\hilbert{H})}.
\end{eqnarray*}
From equation \eqref{honedres} it follows
\begin{equation*}\begin{split}
& \hiepsidres{1}\qm\tilde{\chi}(\hiepsi{0}) = \big(\id-\chi(\hepsisigma)\big)\qless{L}\hiepsi{1}\qless{L}\big(\id-\chi(\hepsisigma)\big)\qm\tilde{\chi}(\hiepsi{0}) = \\
&=\Or(\epsi^{2})_{\mathcal{L}(\hilbert{H})},
\end{split}
\end{equation*}
using lemma \ref{approxchi} twice and lemma \ref{lemmaxi}.

Concerning the second one, applying once again the Duhamel formula, we have
\begin{eqnarray*}
\lefteqn{ -\I\epsi\int_{0}^{t}ds\ \E^{\I(s-t)\hdrestwo/\epsi}(\hiepsidres{2}-h_{2, \mathrm{D}})\E^{-\I s\hdiagtilde/\epsi}\qm\tilde{\chi}(\hiepsi{0})\tilde{\tilde{\chi}}(\hdres)=}\\
&=& -\,\I\epsi\int_{0}^{t}ds\ \E^{\I(s-t)\hdrestwo/\epsi}(\hiepsidres{2}-h_{2, \mathrm{D}})\E^{-\I s\hiepsi{0}/\epsi}\qm\tilde{\chi}(\hiepsi{0})\tilde{\tilde{\chi}}(\hdres) \\
&&+\, \Or(\epsi^{2}\abs{t}^{2}\sqrt{\log(\sigma^{-1})})_{\mathcal{L}(\hilbert{H})}\,,
\end{eqnarray*}
so we have to look at $\E^{\I(s-t)\hdrestwo/\epsi}(\hiepsidres{2} - h_{2, \mathrm{D}})\qm\tilde{\chi}(\hiepsi{0})$. 
Following a procedure already employed several times, we first observe that, in the expression for $\hiepsidres{2}$, equation \eqref{htwodres}, we can replace, making an error of order $\Or(\epsi)$, $\chi(\hepsisigma)$ with $\chi(\hiepsi{0})$.

Applying lemma \ref{lemmaxi}, we see that the last three terms in \eqref{htwodres} vanish, and that the terms containing $[\uonelsigma, \hiepsisigma{1}]$ combine to give the Darwin term and the expression of the effective mass. What remains is exactly $\tilde{h}_{2, \mathrm{OD}}$. Finally, we apply again Duhamel formula to approximate the time evolution generated by $\hdrestwo$, getting
\begin{eqnarray*}
\lefteqn{ \E^{\I(s-t)\hdrestwo/\epsi}(\hiepsidres{2}-h_{2, \mathrm{D}})\qm\tilde{\chi}(\hiepsi{0})=}\\&=&\E^{\I(s-t)h_{0}/\epsi}\tilde{h}_{2, \mathrm{OD}}\qm\tilde{\chi}(\hiepsi{0}) +
\Or(\epsi\abs{t}\sqrt{\log(\sigma^{-1})})\,.
\end{eqnarray*}

We proceed now to show that we can replace $\hdiagtilde$ with $\hdiag$, up to an error of order $\Or(\epsi^{3/2}\abs{t})_{\mathcal{L}(\hilbert{H})}$.

Applying repeatedly Duhamel formula, and putting $\Psi:= \qm\tilde{\chi}(\hdres)\Psi_{0}$, we get
\begin{equation*}\begin{split}
& (\E^{-\I t\hdiagtilde/\epsi} - \E^{-\I t\hdiag/\epsi})\Psi = -\I\epsi\sum_{j=1}^{N}\frac{e_{j}^{2}}{2m_{j}}\int_{0}^{t}ds\, \E^{\I(s-t)\hiepsi{0}/\epsi}a(v_{x_{j}})^{*}a(v_{x_{j}})\E^{-\I s\hiepsi{0}/\epsi}\Psi \\
&+\Or(\epsi^{2}\abs{t}^{2})_{\mathcal{L}(\hilbert{H})}.
\end{split}
\end{equation*}
To streamline the presentation, we assume that $M=1$, the calculations for $M>1$ are basically the same, but more cumbersome.

The integral gives therefore
\begin{equation}\label{estimatetermm1}\begin{split}
&\E^{-\frac{\I}{\epsi} t\hiepsi{0}}f(k_{1}, \lambda_{1})\sum_{\lambda=1}^{2}\int_{\field{R}{3}}dk\, f(k, \lambda)^{*}\int_{0}^{t}ds\, \E^{\I \frac{s}{\epsi}(\abs{k_{1}}-\abs{k})}\E^{\I \frac{s}{\epsi}\hiepsi{\mathrm{p}}}\E^{\I x_{j}\cdot(k-k_{1})}\E^{-\I \frac{s}{\epsi}\hiepsi{\mathrm{p}}}\Psi(k, \lambda)\\
& = \E^{-\frac{\I}{\epsi} t\hiepsi{0}}f(k_{1}, \lambda_{1})\sum_{\lambda=1}^{2}\int_{\field{R}{3}}dk\, \frac{f(k, \lambda)^{*}}{1 + \I(\abs{k_{1}}-\abs{k})\epsi^{-1}}[1 + \I(\abs{k_{1}}-\abs{k})\epsi^{-1}]\cdot\\
&\cdot\int_{0}^{t}ds\, \E^{\I \frac{s}{\epsi}(\abs{k_{1}}-\abs{k})}\E^{\I \frac{s}{\epsi}\hiepsi{\mathrm{p}}}\E^{\I x_{j}\cdot(k-k_{1})}\E^{-\I \frac{s}{\epsi}\hiepsi{\mathrm{p}}}\Psi(k, \lambda).
\end{split}
\end{equation} 
Integrating by parts we get
\begin{equation*}\begin{split}
& \I(\abs{k_{1}}-\abs{k})\epsi^{-1}\int_{0}^{t}ds\, \E^{\I s(\abs{k_{1}}-\abs{k})/\epsi}\E^{\I s\hiepsi{\mathrm{p}}}\E^{\I x_{j}\cdot(k-k_{1})}\E^{-\I s\hiepsi{\mathrm{p}}}\Psi(k, \lambda)=\\
&=\E^{\I t(\abs{k_{1}}-\abs{k})/\epsi}\E^{\I t\hiepsi{\mathrm{p}}}\E^{\I x_{j}\cdot(k-k_{1})}\E^{-\I t\hiepsi{\mathrm{p}}}\Psi - \E^{\I x_{j}\cdot(k-k_{1})}\Psi +\\
&- \frac{\I}{\epsi}\int_{0}^{t}ds\, \E^{\I s(\abs{k_{1}}-\abs{k})/\epsi}\E^{\I s\hiepsi{\mathrm{p}}}[\hiepsi{\mathrm{p}}, \E^{\I x_{j}\cdot(k-k_{1})}]\E^{-\I s\hiepsi{\mathrm{p}}}\Psi,
\end{split}
\end{equation*}
where the commutator is of order $\Or(\epsi)$ when applied to functions of bounded kinetic energy, so that the right-hand side is uniformly bounded in $\epsi$.

We have now to put this expression back in \eqref{estimatetermm1} and estimate the single terms. We show how to do this for the first one, the others being entirely analogous. We ignore the unitary on the left, which does not change the norm, so we have to consider
\begin{equation*}\begin{split}
& f(k_{1}, \lambda_{1})\sum_{\lambda=1}^{2}\int_{\field{R}{3}}dk\, \frac{f(k, \lambda)^{*}}{1 + \I(\abs{k_{1}}-\abs{k})\epsi^{-1}}\cdot\\
&\cdot\int_{0}^{t}ds\, \E^{\I s(\abs{k_{1}}-\abs{k})/\epsi}\E^{\I s\hiepsi{\mathrm{p}}}\E^{\I x_{j}\cdot(k-k_{1})}\E^{-\I s\hiepsi{\mathrm{p}}}\Psi(k, \lambda).
\end{split}
\end{equation*}
Using twice the Cauchy-Schwarz inequality we get
\begin{equation*}\begin{split}
& \norm{\cdots}_{\hilbert{H}}^{2} = \sum_{\sigma_{1}, \ldots, \sigma_{N}, \lambda_{1}=1}^{2}\int_{\field{R}{3N}}dx\int_{\field{R}{3}}dk_{1}\abs{\cdots}^{2}\leq\\
&\leq\abs{t}\sum_{\sigma_{1}, \ldots, \sigma_{N}, \lambda_{1}=1}^{2}\int dx\int dk_{1}\abs{f(k_{1}, \lambda_{1})}^{2}\sum_{\lambda=1}^{2}\int dk \frac{\abs{f(k, \lambda)}^{2}}{1 + (\abs{k_{1}} - \abs{k})^{2}\epsi^{-2}}\cdot\\
&\cdot \sum_{\lambda=1}^{2}\int dk\int_{0}^{t}ds\, \bigg\lvert\E^{\I s\hiepsi{\mathrm{p}}}\E^{\I x_{j}\cdot(k-k_{1})}\E^{-\I s\hiepsi{\mathrm{p}}}\Psi(k, \lambda)\bigg\rvert^{2}=\\
&=\abs{t}^{2}\norm{\Psi}_{\hilbert{H}}^{2}\sum_{\lambda_{1}, \lambda=1}^{2}\int dk_{1}dk\, \frac{\abs{f(k, \lambda)f(k_{1}, \lambda_{1})}^{2}}{1 + (\abs{k_{1}} - \abs{k})^{2}\epsi^{-2}}\leq\\ &\leq C\epsi^{4}\abs{t}^{2}\norm{\Psi}_{\hilbert{H}}^{2}\int_{0}^{\Lambda/\epsi}dk_{1}\int_{0}^{\Lambda/\epsi}dk\, \frac{k_{1}k}{1 + (k_{1} - k)^{2}} = \Or(\epsi\abs{t}^{2}\norm{\Psi}^{2}_{\hilbert{H}}).
\end{split}
\end{equation*}

We proceed now to examine the last three terms in \eqref{htwotildeod} to show that they can be neglected. 

The second and the third term are a sum of terms of the form
\begin{equation*}
\Phi\big(g(\lambda, k)\E^{-\I k\cdot x_{j}}\big)\hat{T},
\end{equation*}
where the function $g\simeq\abs{k}^{\alpha}$, $\abs{k}\to 0^{+}$, with $\alpha = -1/2$ or $+1/2$, and $\hat{T}$ is a Pauli matrix, or the product of two momentum operators. 

For a term of this form we get then, putting again $\Psi:= \qm\tilde{\chi}(\hdres)\Psi_{0}$, 
\begin{equation*}\begin{split}
& \I\epsi\int_{0}^{t}ds\ \E^{\I(s-t)\hiepsi{0}/\epsi}\Phi\big(g(\lambda, k)\E^{-\I k\cdot x_{j}}\big)\hat{T}\E^{-\I s\hiepsi{0}/\epsi}\qm\tilde{\chi}(\hdres)\Psi_{0}=\\
&= \I\frac{\epsi}{\sqrt{2}}\int_{0}^{t}ds\ \E^{\I(s-t)\hiepsi{0}/\epsi}[a(g\E^{-\I k\cdot x_{j}}) + a(g\E^{-\I k\cdot x_{j}})^{*}]\hat{T}\E^{-\I s\hiepsi{0}/\epsi}\Psi\quad .
                 \end{split}
\end{equation*}
Expression of this type have already been estimated in (\cite{TeTe}, theorem $4$). For convenience of the reader, we give the proof for the annihilation part, referring to \cite{TeTe} for the creation part, which is entirely analogous.  

The annihilation part gives
\begin{equation}\label{annihilationpart}\begin{split}
& \I\frac{\epsi}{\sqrt{2}}\E^{-\I t\hiepsi{0}/\epsi}\sqrt{M}\int_{0}^{t}ds\, \E^{\I s\hiepsi{\mathrm{p}}/\epsi} \sum_{\lambda=1}^{2}\int_{\field{R}{3}}dk\,g(k, \lambda)^{*}\E^{\I k\cdot x_{j}}\hat{T}\E^{-\I s\abs{k}/\epsi}\E^{-\I s\hiepsi{\mathrm{p}}/\epsi}\cdot\\
&\cdot\Psi(x, \sigma; k, \lambda, k_{1}, \lambda_{1}, \ldots, k_{M}, \lambda_{M}) =\\
&= \I\frac{\epsi}{\sqrt{2}}\E^{-\I t\hiepsi{0}/\epsi}\sqrt{M}\sum_{\lambda=1}^{2}\int_{\field{R}{3}}dk\,\frac{g(k, \lambda)^{*}}{1 - \I\abs{k}\epsi^{-1}}(1 - \I\abs{k}\epsi^{-1})\cdot\\
&\cdot\int_{0}^{t}ds\, \E^{-\I s\abs{k}/\epsi}\E^{\I s\hiepsi{\mathrm{p}}/\epsi}\E^{\I k\cdot x_{j}}\hat{T}\E^{-\I s\hiepsi{\mathrm{p}}/\epsi}\Psi(x, \sigma; k, \lambda, \ldots) \quad .
\end{split}
\end{equation}
Integrating by parts we get
\begin{equation*}\begin{split}
& - \I\abs{k}\epsi^{-1}\int_{0}^{t}ds\, \E^{-\I s\abs{k}/\epsi}\E^{\I s\hiepsi{\mathrm{p}}/\epsi}\E^{\I k\cdot x_{j}}\hat{T}\E^{-\I s\hiepsi{\mathrm{p}}/\epsi}\Psi(x, \sigma; k, \lambda, \ldots)=\\
& = \E^{-\I t\abs{k}/\epsi}\E^{\I t\hiepsi{\mathrm{p}}/\epsi}\E^{\I k\cdot x_{j}}\hat{T}\E^{-\I t\hiepsi{\mathrm{p}}/\epsi}\Psi - \E^{\I k\cdot x_{j}}\hat{T}\Psi +\\
&-\frac{\I}{\epsi}\int_{0}^{t}ds\, \E^{-\I s\abs{k}/\epsi}\E^{\I s\hiepsi{\mathrm{p}}/\epsi}[\hiepsi{\mathrm{p}}, \E^{\I k\cdot x_{j}}\hat{T}]\E^{-\I s\hiepsi{\mathrm{p}}/\epsi}\Psi,
\end{split}
\end{equation*}
where
\begin{equation*}\begin{split}
& [\hiepsi{\mathrm{p}}, \E^{\I k\cdot x}\hat{T}] = \frac{1}{2m_{j}}(2\epsi \E^{\I k\cdot x_{j}}k\cdot\pop_{j} + \epsi^{2}\abs{k}^{2}\E^{\I k\cdot x_{j}})\hat{T} + \E^{\I k\cdot x_{j}}[\vcoul, \hat{T}]\quad.
\end{split}
\end{equation*}
The commutator on the right-hand side is zero if $\hat{T}$ is a Pauli matrix, or it is a term of order $\epsi$ times $\pop$, if $\hat{T}$ is the product of two momentum operators. Therefore it has the same form as the first part of the right-hand side and can be treated in the same way.

We have now to put the result of the integration by parts back in equation \eqref{annihilationpart} and estimate what comes out. We show how to do this for the first term, all the other ones can be treated in the same way.

Ignoring the unitary on the left and the constants, we consider then
\begin{equation*}\begin{split}
&\epsi\sum_{\lambda=1}^{2}\int_{\field{R}{3}}dk\,\frac{g(k, \lambda)^{*}}{1 - \I\abs{k}\epsi^{-1}}\int_{0}^{t}ds\, \E^{-\I s\abs{k}/\epsi}\E^{\I s\hiepsi{\mathrm{p}}/\epsi}\E^{\I k\cdot x_{j}}\hat{T}\E^{-\I s\hiepsi{\mathrm{p}}/\epsi}\Psi(x, \sigma; k, \lambda, \ldots).
\end{split}
\end{equation*}
Using the Cauchy-Schwarz inequality we get
\begin{equation*}\begin{split}
& \abs{\cdots}^{2}\leq \epsi^{2}\sum_{\lambda=1}^{2}\int_{\field{R}{3}}dk\, \frac{\abs{g(k, \lambda)}^{2}}{1+\abs{k}^{2}\epsi^{-2}}\cdot\\
&\cdot\abs{t}\sum_{\lambda=1}^{2}\int_{\field{R}{3}}dk\int_{0}^{t}ds\, \abs{\E^{\I s\hiepsi{\mathrm{p}}/\epsi}\E^{\I k\cdot x_{j}}\hat{T}\E^{-\I s\hiepsi{\mathrm{p}}/\epsi}\Psi}^{2},
\end{split}
\end{equation*}
so
\begin{equation*}\begin{split}
& \norm{\cdots}_{\hilbert{H}}^{2}\leq \epsi^{2}\sum_{\lambda=1}^{2}\int_{\field{R}{3}}dk\, \frac{\abs{g(k, \lambda)}^{2}}{1+\abs{k}^{2}\epsi^{-2}}\abs{t}\int_{0}^{t}ds\,\norm{\hat{T}\E^{-\I s\hiepsi{\mathrm{p}}/\epsi}\Psi}^{2}_{\hilbert{H}}\quad .
\end{split}
\end{equation*}
The left integral gives
\begin{equation*}
\int_{\field{R}{3}}dk\, \frac{\abs{g(k, \lambda)}^{2}}{1 + \abs{k}^{2}\epsi^{-2}}\leq C\int_{0}^{\Lambda}d\abs{k}\frac{\abs{k}^{2(1-\alpha)}}{1 + \abs{k}^{2}\epsi^{-2}} = \Or\big(\epsi^{2}\log(\Lambda\epsi^{-1})\big),
\end{equation*}
when $\alpha = 1/2$, or $-1/2$.

The same analysis can be carried out for the remaining term
\begin{equation*}
\sum_{j=1}^{N}\frac{e_{j}^{2}}{2m_{j}}\big[a(v_{x_{j}})^{2} + a(v_{x_{j}})^{*\,2}] \quad.
\end{equation*}
The proof is identical to the one given in (\cite{TeTe}, theorem $4$) for a similar term appearing in the case of Nelson model, and depends as above on the fact that $v_{x_{j}}(\lambda, k)\simeq\abs{k}^{-1/2}$, $\abs{k}\to 0^{+}$.
\end{proof}

\begin{corollary}\label{radiatedpiece} At the leading order, the radiated piece (i.\ e.\ the piece of the wave function which makes a transition between the almost invariant subspaces) for a system starting in the Fock vacuum, $\Psi_{0}(x)=\psi(x)\Omega_{\mathrm{F}}$, $\psi(x)\in\hilbertp$, is given by
\begin{equation}\label{equationtt}
-\E^{-\I t \hiepsi{0}}\frac{\I\epsi}{\sqrt{2}} \frac{\hat{\varphi}_{\sigma(\epsi)}(k)}{\abs{k}^{3/2}}e_{\lambda}(k)\cdot\int_{0}^{t}ds\, \E^{\I(s-t)\abs{k}/\epsi}\opw\bigg(\sum_{j=1}^{N}\frac{e_{j}}{m_{j}}\ddot{x}_{j}^{cl}(s; x, p)\bigg)\psi(x)\, ,
\end{equation}
where $x_{j}^{cl}$ is the solution to the classical equations of motion
\begin{equation}\begin{split}
& m_{j}\ddot{x}_{j}^{cl}(s; x, p) = -\nabla_{x_{j}}\vcoul(x^{cl}(s; x, p)),\\
& x_{j}^{cl}(0; x, p)= x_{j},\qquad \dot{x}_{j}^{cl}(0; x, p)= p_{j}m_{j}^{-1},\quad j=1, \ldots, N\, .
\end{split}
\end{equation}

This coincides with the leading order of the radiated piece corresponding to the original Hamiltonian $\hepsi$, for a system starting in the dressed vacuum, $\uepsi^{*}\omegaf$.
\end{corollary}

\begin{proof}
Applying equation \eqref{firstordertime} for the case $M=0$ we get at the leading order
\begin{eqnarray*}
\lefteqn{Q_{0}^{\perp}\E^{-\I t\hdres/\epsi}\psi(x)\omegaf=-\I\epsi\int_{0}^{t}ds\ \E^{\I(s-t)\hiepsi{0}/\epsi}h_{2, \mathrm{OD}}\E^{-\I s\hiepsi{0}/\epsi}\psi(x)\Omega_{\mathrm{F}}=}\\
&=&-\frac{\epsi}{\sqrt{2}}\sum_{j=1}^{N}\frac{e_{j}}{m_{j}}\int_{0}^{t}ds\, \E^{\I(s-t)\hp/\epsi}\E^{\I(s-t)\abs{k}/\epsi}\nabla_{x_{j}}\vcoul\cdot \frac{\I v_{x_{j}, \sigma}(k, \lambda)}{\abs{k}}\E^{-\I s\hp/\epsi}\psi(x)\\
&=&-\frac{\I\epsi}{\sqrt{2}} \frac{\hat{\varphi}_{\sigma(\epsi)}(k)}{\abs{k}^{3/2}}e_{\lambda}(k)\cdot\sum_{j=1}^{N}\frac{e_{j}}{m_{j}}\int_{0}^{t}ds\, \E^{\I(s-t)\abs{k}/\epsi}\E^{\I (s-t)\hp/\epsi}(\E^{-\I k\cdot x_{j}}-1)\\
&&\qquad\qquad\qquad\qquad\qquad\qquad\qquad\quad\cdot\nabla_{x_{j}}\vcoul\E^{-\I s\hp/\epsi}\psi(x)\\
&&-\frac{\I\epsi}{\sqrt{2}} \frac{\hat{\varphi}_{\sigma(\epsi)}(k)}{\abs{k}^{3/2}}e_{\lambda}(k)\cdot\sum_{j=1}^{N}\frac{e_{j}}{m_{j}}\int_{0}^{t}ds\, \E^{\I(s-t)\abs{k}/\epsi}\E^{\I (s-t)\hp/\epsi}\nabla_{x_{j}}\vcoul\E^{-\I s\hp/\epsi}\\
&&\qquad\qquad\qquad\qquad\qquad\qquad\qquad\quad\cdot\psi(x)\\
&=&-\frac{\I\epsi}{\sqrt{2}} \frac{\hat{\varphi}_{\sigma(\epsi)}(k)}{\abs{k}^{3/2}}e_{\lambda}(k)\cdot\sum_{j=1}^{N}\frac{e_{j}}{m_{j}}\int_{0}^{t}ds\, \E^{\I(s-t)\abs{k}/\epsi}\E^{\I (s-t)\hp/\epsi}(\E^{-\I k\cdot x_{j}}-1)\\
&&\qquad\qquad\qquad\qquad\qquad\qquad\qquad\quad\cdot\nabla_{x_{j}}\vcoul\E^{-\I s\hp/\epsi}\psi(x)\\
&&-\,\E^{-\I t \hiepsi{0}}\frac{\I\epsi}{\sqrt{2}} \frac{\hat{\varphi}_{\sigma(\epsi)}(k)}{\abs{k}^{3/2}}e_{\lambda}(k)\cdot\int_{0}^{t}ds\, \E^{\I(s-t)\abs{k}/\epsi}\opw\bigg(\sum_{j=1}^{N}\frac{e_{j}}{m_{j}}\ddot{x}_{j}^{cl}(s; x, p)\bigg)\psi(x)\\
&&+\,\Or(\epsi^{2}\abs{t})_{\mathcal{L}(\hilbert{H})}\norm{\psi}_{L^{2}(\field{R}{3n})}\,,
\end{eqnarray*}
where we have used Egorov's theorem to approximate $\E^{\I s\hp/\epsi}\nabla_{x_{j}}\vcoul\E^{-\I s\hp/\epsi}$ (see, e. g., \cite{Ro}).

To end the proof of the first statement we have to show that the norm of the first term is small. For the $j$th term in the sum we get
\begin{eqnarray*}
\lefteqn{ \norm{\cdots}^{2}_{\hilbert{H}}\leq \epsi^{2}\frac{e_{j}^{2}}{m_{j}^{2}}\sum_{spin}\int dx\, \int dk\, \frac{\abs{\hat{\varphi}_{\sigma(\epsi)}(k)}^{2}}{\abs{k}^{3}}\bigg\lvert\int_{0}^{t}ds\, \E^{\I s\abs{k}/\epsi}\E^{\I s\hp/\epsi}(\E^{\I k\cdot x_{j}}-1)\cdot}\\
&&\cdot\nabla_{x_{j}}\vcoul(x)\E^{-\I s\hp/\epsi}\psi(x)\bigg\rvert^{2}\\
&=&\epsi^{2}\frac{e_{j}^{2}}{m_{j}^{2}}\sum_{spin}\int dx\, \int dk\, \frac{\abs{\hat{\varphi}_{\sigma(\epsi)}(\epsi k)}^{2}}{\abs{k}^{3}}\bigg\lvert\int_{0}^{t}ds\, \E^{\I s\abs{k}}\E^{\I s\hp/\epsi}(\E^{\I \epsi k\cdot x_{j}}-1)\cdot\\
&&\cdot\nabla_{x_{j}}\vcoul(x)\E^{-\I s\hp/\epsi}\psi(x)\bigg\rvert^{2}\\
&\leq& \epsi^{2}\abs{t}\frac{e_{j}^{2}}{m_{j}^{2}} \int dk\, \frac{\abs{\hat{\varphi}_{\sigma(\epsi)}(\epsi k)}^{2}}{\abs{k}^{3}}\int_{0}^{t}ds\, \bigg\lVert(\E^{\I \epsi k\cdot x_{j}}-1)\nabla_{x_{j}}\vcoul(x)\E^{-\I s\hp/\epsi}\psi(x)\bigg\rVert^{2}_{\hilbertp}\\
&\leq& C\epsi^{4}\abs{t} \int_{\sigma(\epsi)/\epsi}^{\Lambda/\epsi} d\abs{k}\, \frac{1}{\abs{k}}\int_{0}^{t}ds\, \bigg\lVert x_{j}\nabla_{x_{j}}\vcoul(x)\E^{-\I s\hp/\epsi}\psi(x)\bigg\rVert^{2}_{\hilbertp}\\
&=&C\epsi^{4}\abs{t} \log\bigg(\frac{\Lambda}{\sigma(\epsi)}\bigg)\int_{0}^{t}ds\, \bigg\lVert x_{j}\nabla_{x_{j}}\vcoul(x)\E^{-\I s\hp/\epsi}\psi(x)\bigg\rVert^{2}_{\hilbertp}\, .
\end{eqnarray*}

For the second statement we have
\begin{equation*}\begin{split}
& (\id - P_{0}^{\epsi})\E^{-\frac{\I}{\epsi} t\hepsi}P_{0}^{\epsi}\chi(\hepsi)\Psi = (\id - P_{0}^{\epsi})(\E^{-\frac{\I}{\epsi} t\hepsi} - \E^{-\frac{\I}{\epsi} t\hepsisigma})P_{0}^{\epsi}\chi(\hepsi)\Psi +\\
&+ (\id - P_{0}^{\epsi})\E^{-\frac{\I}{\epsi} t\hepsi}P_{0}^{\epsi}[\chi(\hepsi)-\chi(\hepsisigma)]\Psi + (\id - P_{0}^{\epsi})\E^{-\frac{\I}{\epsi} t\hepsisigma}P_{0}^{\epsi}\chi(\hepsisigma)\Psi =\\
&= \Or\big(\sigma(\epsi)^{1/2}\abs{t}\norm{\Psi}_{\hilbert{H}}\big) + \hilbert{U}^{*}Q_{0}^{\perp}\E^{-\I t\hdres/\epsi}\psi(x)\omegaf\hilbert{U} =\\
&= Q_{0}^{\perp}\E^{-\I t\hdres/\epsi}\psi(x)\omegaf + \Or\big(\epsi^{2}\log\sigma(\epsi)^{-1}\norm{\Psi}_{\hilbert{H}}\big) + \Or\big(\sigma(\epsi)\abs{t}\norm{\Psi}_{\hilbert{H}}\big)\quad ,
\end{split}
\end{equation*}
with
\begin{equation*}
\psi(x) := <\omegaf, \chi(\hiepsi{0})\Psi>_{\fock}\, \in\hilbertp\quad .
\end{equation*}
\end{proof}

\begin{remark}\label{remarkradiatedpiece}
Denoting by 
\begin{equation*}
T_{t}(k) := \opw\bigg(\int_{0}^{t}ds\, \E^{\I(s-t)k}\sum_{j=1}^{N}\frac{e_{j}}{m_{j}}\ddot{x}_{j}^{cl}(s; x, p)\bigg),
\end{equation*}
the operator acting on $\hilbertp$ which appears in \eqref{equationtt}, the norm squared of the leading part of the radiated piece is 
\begin{equation*}\begin{split}
& \frac{\epsi^{2}}{2}\sum_{spin}\sum_{\lambda=1, 2}\int dx\int dk\, \frac{\abs{\hat{\varphi}_{\sigma(\epsi)}(k)}^{2}}{\abs{k}^{3}}\abs{e_{\lambda}(k)\cdot T_{t}(\abs{k}/\epsi)\psi(x)}^{2} =\\
&= \frac{4}{3}\pi\epsi^{2}\int d\abs{k}\, \frac{\abs{\hat{\varphi}_{\sigma(\epsi)}(\abs{k})}^{2}}{\abs{k}}\norm{T_{t}(\abs{k}/\epsi)\psi}_{\hilbertp\otimes\field{C}{3}}^{2} =\\
&= \frac{\epsi^{2}}{6\pi^{2}}\int_{\sigma(\epsi)\epsi^{-1}}^{\Lambda\epsi^{-1}}d\abs{k}\,\frac{1}{\abs{k}}\norm{T_{t}(\abs{k})\psi}_{\hilbertp\otimes\field{C}{3}}^{2} \geq \\ &\geq \frac{\epsi^{2}}{6\pi^{2}}\int_{\sigma(\epsi)\epsi^{-1}}^{a}d\abs{k}\,\frac{1}{\abs{k}}\norm{T_{t}(\abs{k})\psi}_{\hilbertp\otimes\field{C}{3}}^{2},
\end{split}
\end{equation*}
where $a>0$ can be chosen arbitrarily small. 

The symbol of $T_{t}(k)$ is an $\epsi$-independent function, which for $k=0$ is different from the null function
\begin{equation*}
\int_{0}^{t}ds\, \sum_{j=1}^{N}\frac{e_{j}}{m_{j}}\ddot{x}_{j}^{cl}(s; x, p) = \sum_{j=1}^{N}\frac{e_{j}}{m_{j}}[\dot{x}_{j}^{cl}(t; x, p) - p_{j}], 
\end{equation*}
so $T_{t}(0)$ is different from the zero operator. We expect therefore that for a generic state $\psi$
\begin{equation*}
\inf_{0<\abs{k}<a}\norm{T_{t}(\abs{k})\psi}_{\hilbertp\otimes\field{C}{3}} >0 \quad . 
\end{equation*}
In this case one gets as lower bound for the norm of the radiated piece
\begin{equation*}
\inf_{0<\abs{k}<a}\norm{T_{t}(\abs{k})\psi}_{\hilbertp\otimes\field{C}{3}}\frac{\epsi}{\sqrt{6}\pi}\bigg(\int_{\sigma(\epsi)\epsi^{-1}}^{a}d\abs{k}\,\frac{1}{\abs{k}}\bigg)^{1/2} = \Or\big(\epsi\sqrt{\log(\epsi\sigma(\epsi)^{-1}})\big),
\end{equation*}
which is almost of the same order as the upper bound.
\end{remark}

\begin{remark}\label{remradiatedpower}
The radiated energy, defined in equation \eqref{radiatedenergy}, can be written at the leading order as
\begin{equation*}
E_{\textrm{rad}}(t)\cong \langle Q_{0}^{\perp}\E^{-\I t\hdres/\epsi}\psi(x)\omegaf, \Hf Q_{0}^{\perp}\E^{-\I t\hdres/\epsi}\psi(x)\omegaf\rangle,
\end{equation*}
where $\psi$ is defined in \eqref{psi}.
Using the expression for the radiated piece we get  
\begin{equation*}\begin{split}
E_{\mathrm{rad}}(t) \cong &\frac{\epsi^{2}}{2}\sum_{spin}\sum_{\lambda = 1, 2}\int dx \int dk\, \frac{\abs{\hat{\varphi}_{\sigma(\epsi)}(k)}^{2}}{\abs{k}^{2}}\abs{e_{\lambda}(k)\cdot T_{t}(\abs{k}/\epsi)\psi(x)}^{2}=\\
&=\frac{4}{3}\pi\epsi^{2}\int d\abs{k} \abs{\hat{\varphi}_{\sigma(\epsi)}(\abs{k})}^{2}\norm{T_{t}(\abs{k}/\epsi\psi(x)}^{2}_{\hilbertp\otimes\field{C}{3}}=\\
&\cong \frac{4}{3}\pi\epsi^{2}\int d\abs{k} \abs{\hat{\varphi}(\abs{k})}^{2}\int_{0}^{t}ds\int_{0}^{t}ds'\, \E^{\frac{\I}{\epsi}(s-s')\abs{k}}\cdot\\
&\qquad\qquad\cdot\sum_{j, l=1}^{N}\frac{e_{j}e_{l}}{m_{j}m_{l}}\langle \psi(x), \opw\big(\ddot{x}_{j}^{cl}(s')\cdot\ddot{x}_{l}^{cl}(s)\big)\psi(x)\rangle_{\hilbertp}=\\
&= \frac{\epsi^{2}}{6\pi^{2}}\int_{0}^{t}ds\int_{0}^{t}ds'\, \frac{\epsi}{\I(s-s')}(\E^{\frac{\I}{\epsi}(s-s')\Lambda} - 1)\cdot\\
&\qquad\qquad\qquad\qquad\quad\cdot\langle\psi(x), \opw(\ddot{D}(s')\cdot\ddot{D}(s))\psi(x)\rangle_{\hilbertp},
\end{split}
\end{equation*}
where we have used the product formula for pseudodifferential operators (see, e. g., \cite{Ro}) and defined
\begin{equation*}
D(s; x, p) := \sum_{j=1}^{N}\frac{e_{j}}{m_{j}}x_{j}^{cl}(s; x, p)\quad .
\end{equation*}
The radiated power is then
\begin{equation*}\begin{split}
& P_{\textrm{rad}}(t) = \frac{d}{dt}E_{\textrm{rad}}(t) \cong \frac{\epsi^{3}}{3\pi^{2}}\int_{0}^{t}ds\, \frac{\sin[(t-s)\Lambda/\epsi]}{t-s}\langle \psi, \opw\big(\ddot{D}(s')\cdot\ddot{D}(s)\big)\psi\rangle_{\hilbertp}\\
&\overset{\epsi\to 0}{\cong} \frac{\epsi^{3}}{3\pi^{2}}\langle \psi, \opw\big(\abs{\ddot{D}(t)}^{2}\big)\psi\rangle_{\hilbertp} \quad .
\end{split}
\end{equation*} 
\end{remark}

\begin{corollary} Let 
\begin{equation*}
\omega(t):= \E^{-\I t\hdres/\epsi}\omega_{0}\E^{\I t\hdres/\epsi},
\end{equation*}
where $\omega_{0}\in\hilbert{I}_{1}(\qm\tilde{\chi}(\hdres)\hilbert{H})$, the Banach space of trace class operators on $\qm\tilde{\chi}(\hdres)\hilbert{H}$, and let $\omegap$ be the partial trace over the field states
\begin{equation*}
\omegap(t):= \tr{\fock}\,\omega(t),
\end{equation*}
then
\begin{eqnarray*} 
 \omegap(t) &=& \E^{-\I t\hdiagp/\epsi}\omegap(0)\E^{\I t\hdiagp/\epsi}+ \\&&+\, \Or(\epsi^{3/2}\abs{t})_{\hilbert{I}_{1}(\hilbertp)}(1 - \delta_{M0})  + \Or\big(\epsi^{2}\sqrt{\log(\sigma(\epsi)^{-1})}(\abs{t} + \abs{t}^{2})\big)_{\hilbert{I}_{1}(\hilbertp)}  ,\nonumber
\end{eqnarray*}
where
\begin{equation*}\begin{split}
\hdiagp &:= \sum_{j=1}^{N} \frac{1}{2m_{j}}\pop_{j}^{2} + \vcoul + \\
&-\epsi^{2}\sum_{l, j=1}^{N}\frac{e_{j}e_{l}}{m_{j}m_{l}}\int_{\field{R}{3}}dk\, \frac{\abs{\hat{\varphi}(k)}^{2}}{2\abs{k}^{2}}\E^{\I k\cdot x_{j}}\pop_{j}\cdot(\id - \kappa\otimes\kappa)\pop_{l}\E^{-\I k\cdot x_{l}}\quad .
\end{split}
\end{equation*}
and $\hilbert{I}_{1}(\hilbertp)$ denotes the space of trace class operators on $\hilbertp$. 
\end{corollary}

\begin{proof}
The proof follows from the following three facts:
\begin{enumerate}
\item the term of order $\epsi$ in equation \eqref{firstordertime} is off-diagonal with respect to the $\qm$;

\item the diagonal Hamiltonian $\hdiag$, defined in \eqref{hdiag}, is equal to $\hdiagp\otimes\mathbf{1}+\mathbf{1}\otimes\Hf$, so we have that
\begin{equation*}
\mathrm{tr}_{\fock}\big(\E^{-\I t\hdiag/\epsi}\omega_{0}\E^{\I t\hdiag/\epsi}\big)=\E^{-\I t\hdiagp/\epsi}\mathrm{tr}_{\fock}(\omega_{0})\E^{\I t\hdiagp/\epsi}\, ;
\end{equation*}

\item the following well known inequality,which holds for any Hilbert space $\hilbert{H}$:
\begin{equation*}
\norm{AB}_{\hilbert{I}_{1}(\hilbert{H})}\leq \norm{A}_{\hilbert{I}_{1}(\hilbert{H})}\cdot\norm{B}_{\mathcal{L}(\hilbert{H})}\,.
\end{equation*}

\end{enumerate}
\end{proof}

\begin{theorem}\label{densitymatrix}
Let $S$ be an observable for the particles, $S\in\mathcal{L}(\hilbertp)$, and $\omega\in \hilbert{I}_{1}(P_{M}^{\epsi}\chi(\hepsi)\hilbert{H})$ a density matrix for a mixed dressed state with $M$ free photons whose time evolution is defined by
\begin{equation}
\omega(t):= e^{-\I t\hepsi/\epsi}\omega e^{\I t\hepsi/\epsi},
\end{equation}
then
\begin{equation}\begin{split}\label{trace}  \tr{\hilbert{H}}\big((S\otimes\id_{\fock})\omega(t)\big)& = \tr{\hilbertp}\big(S e^{-\I t H^{(2)}_{\mathrm{D, p}}}\tr{\fock}(\omega) e^{\I t H^{(2)}_{\mathrm{D, p}}}\big) +\\ &+\Or(\epsi^{3/2}\abs{t})(1 - \delta_{M0}) + \Or\big(\epsi^{2}\sqrt{\log(\sigma(\epsi)^{-1})}(\abs{t} + \abs{t}^{2})\big) .
\end{split}
\end{equation}
\end{theorem}

\begin{proof}
First of all we observe that, using proposition \ref{infraredcutoff} and lemma \ref{lemmachiinfra}, we have
\begin{equation*}\begin{split}
& \tr{\hilbert{H}}\bigg(\big(S\otimes\id_{\fock}\big)\omega(t)\bigg) = \tr{\hilbert{H}}\bigg(\big(S\otimes\id_{\fock}\big)e^{-\I t\hepsisigma/\epsi}\omega_{\sigma(\epsi)}\cdot\\
&\cdot e^{\I t\hepsisigma/\epsi}\bigg) + \Or(\sigma(\epsi)^{1/2})\,,
\end{split}
\end{equation*}
where $\omega_{\sigma(\epsi)}\in \hilbert{I}_{1}(P_{M}^{\epsi}\chi(\hepsisigma)\hilbert{H})$.
By the definition of the dressed Hamiltonian and the cyclicity of the trace we have then at the leading order
\begin{eqnarray*} \lefteqn{\hspace{-5mm}
 \tr{\hilbert{H}}\bigg(\big(S\otimes\id_{\fock}\big)\omega(t)\bigg)\simeq}\\&\simeq& \tr{\hilbert{H}}\bigg(\hilbert{U}\big(S\otimes\id_{\fock}\big)\hilbert{U}^{*}e^{-\I t\hdres/\epsi}\hilbert{U}\omega_{\sigma(\epsi)}  \hilbert{U}^{*}e^{\I t\hdres/\epsi}\bigg) \, .
\end{eqnarray*}
The transformed observable, using the definition of $\hilbert{U}$ and lemma \ref{approxchi}, is given by
\begin{equation*}\begin{split} 
\hilbert{U}\big(S\otimes\id_{\fock}\big)\hilbert{U}^{*} = &S\otimes\id_{\fock} +\epsi\chi(\hiepsi{0})\uonelsigma\big(S\otimes\id_{\fock}\big) +\\
&+\epsi\big(\id - \chi(\hiepsi{0})\big)\uonelsigma\chi(\hiepsi{0})\big(S\otimes\id_{\fock}\big) + \\
&-\epsi\big(S\otimes\id_{\fock}\big)\chi(\hiepsi{0})\uonelsigma +\\& -\epsi\big(S\otimes\id_{\fock}\big)\big(\id - \chi(\hiepsi{0})\big)\uonelsigma\chi(\hiepsi{0}) +\\ &+\Or(\epsi^{2}\log(\sigma^{-1})) \, .
\end{split}
\end{equation*}

All the terms of order $\epsi$ in the previous expression are off-diagonal with respect to the $\qm$s, and the same holds for the term of order $\epsi$ in \eqref{firstordertime}. Therefore, they all vanish when we calculate the trace. Using point $2$ and $3$ of last corollary we get then \eqref{trace} with 
\begin{equation*}
\hilbert{U}\omega_{\sigma(\epsi)}\hilbert{U}^{*}
\end{equation*}
instead of $\omega$.

Using again lemma \ref{lemmachiinfra} and the fact that the terms of order $\epsi$ in the expansion of $\hilbert{U}$ are off-diagonal we can in the end replace $\hilbert{U}$ by the identity and $\omega_{\sigma(\epsi)}$ by $\omega$.
\end{proof}

\appendix
\section{The limit $c\to\infty$}\label{appendix}

In this appendix we sketch the proof of theorem \ref{ctoinfinity}. The reader can find additional discussions in \cite{Da1} and (\cite{Sp}, chapter $17$ and section $20.2$).

As remarked in the introduction, see equations \eqref{equationhlambda}-\eqref{hlambdafour}, the limit $c\to\infty$ has the form of a weak coupling limit, in which the weak interaction is observed over the long time scale $\tau = c^{2}t$. The corresponding physical interpretation is that the small system made up of the particles interacts with an environment (the quantized field) which is traced out to analyze the dynamics of the small system only.

The mathematical framework, as explained in \cite{Da1}, whose notation is employed in this appendix, considers the Banach space
\begin{equation*}
\B:= \hilbert{I}_{1}(\hilbert{H}),
\end{equation*}
the space of trace class operators on $\hilbert{H}$ with the corresponding trace norm. 

$\B$ contains the convex subset of positive operators of trace one, which are the mixed states of the composite system particles plus field (density matrices). 

The Hamiltonian $\hlambda$ defines on $\B$ what in semigroup theory is called an ``implemented semigroup''\footnote[1]{In this case an implemented group of isometries.} (\cite{Al} and references therein) via the usual formula for the Schr\"odinger evolution of the states
\begin{equation*}
V_{t}^{\lambda}(\varrho) := \E^{-\I t\hlambda}\varrho\E^{+\I t\hlambda} \quad.
\end{equation*}
Contrary to what in general can happen for an implemented semigroup, $V_{t}^{\lambda}$ is strongly continuous on $\B$ (\cite{Mo}, theorem $2$) and can therefore be written as 
\begin{equation*}
V_{t}^{\lambda} = \E^{-\I t\llambda},
\end{equation*} 
where $\llambda$ is the generator of $V_{t}^{\lambda}$, called the (total) \emph{Liouvillean}.

In the same way we define 
\begin{eqnarray*}
U_{t}(\varrho) &:=& \E^{-\I t\Hf}\varrho\E^{+\I t\Hf}, \qquad U_{t} = \E^{-\I t\lf},\\
T_{t}^{\lambda}(\varrho) &:=& \E^{-\I t(\Hf + \lambda^{2/3}\Hp)}\varrho\E^{+\I t(\Hf + \lambda^{2/3}\Hp)}, \qquad T_{t}^{\lambda} = \E^{-\I t(\lf + \lambda^{2/3}\lp)}.
\end{eqnarray*}
and the Liouvilleans $\lone$ and $\lfour$ associated to the interaction Hamiltonians $\hilambda{1}$ and $\hilambda{4/3}$.

If the Hamiltonian $H$ which implements the group is a bounded operator, the corresponding Liouvillean ${\sf L}$ is also bounded and given by
\begin{equation*}
{\sf L}(\varrho) = [H, \varrho], \qquad \forall\varrho\in\B\quad.
\end{equation*}
When $H$ is unbounded a little more care is needed, because ${\sf L}$ will be also unbounded and defined only on a suitable dense domain (see, e. g., \cite{PrTi}).

One is interested in studying the dynamics of the particles in the limit $\lambda\to 0$, given that at time $t=0$ the field is in the reference state $\omegar\in\hilbert{I}_{1}(\fock)$, which is invariant under the free dynamics $U_{t}$. Under these conditions a natural choice is $\omegar = Q_{0}$, the projector on the Fock vacuum. Contrary to what happens for the limit $\epsi\to 0$, it is not possible here to look at the case $M\neq 0$, because there does not exist any density matrix in $\qm\fock$ which commutes with $\Hf$. 

To implement these ideas mathematically one defines a projection on the particles states,
\begin{equation*}
\pzeroel(\varrho) := \chi_{(-\infty, \mathrm{E})}(\Hp^{1/2})\tr{\fock}(\qless{L}\varrho\qless{L})\chi_{(-\infty, \mathrm{E})}(\Hp^{1/2})\otimes \omegar,
\end{equation*}
where $\chi$ is the characteristic function of the interval indicated and $\mathrm{L}$ is a fixed integer greater than $2$.

\begin{lemma*}
$\pzeroel$ is a projection of norm $1$.
\end{lemma*}
In the following we denote for simplicity $\pzeroel$ simply by $P_{0}$. Moreover we put $P_{1}:= \id - P_{0}$ and, given an operator $A$, we denote by $A^{(ij)} := P_{i}AP_{j}$.

The dynamics of the particles are given then by
\begin{equation*}
W_{t}^{\lambda} :=  P_{0}V_{t}^{\lambda}P_{0} \quad .
\end{equation*}

We apply now a standard procedure to get an integral equation for $W_{t}^{\lambda}$, called \emph{generalized master equation} (see \cite{Da1} or, in a more general context, \cite{AlLe}, section $\mathrm{III}.1$).

We denote by $U_{t}^{\lambda}$ the diagonal evolution
\begin{equation*}
U_{t}^{\lambda} := \E^{-\I t(\lf + \lambda\lone^{(11)} + \lambda^{4/3}\lfour^{(11)})}.
\end{equation*}

Applying Duhamel Formula and noting that $\lone^{(00)} = \lfour^{(00)} = 0$, one gets 
\begin{equation*}
V_{t}^{\lambda} = U_{t}^{\lambda} + \lambda\int_{0}^{t}ds\, U_{t-s}^{\lambda}(\lone^{(01)} + \lone^{(10)})V_{s}^{\lambda} + \lambda^{4/3}\int_{0}^{t}ds\, U_{t-s}^{\lambda}(\lfour^{(01)} + \lfour^{(10)})V_{s}^{\lambda} \quad .
\end{equation*}

Then
\begin{equation*}
W_{t}^{\lambda} = P_{0}U_{t}^{\lambda} + \lambda\int_{0}^{t}ds\, U_{t-s}^{\lambda}\lone^{(01)}P_{1}V_{s}^{\lambda}P_{0} + \lambda^{4/3}\int_{0}^{t}ds\, U_{t-s}^{\lambda}\lfour^{(01)}P_{1}V_{s}^{\lambda}P_{0},
\end{equation*}
and
\begin{equation*}
P_{1}V_{t}^{\lambda}P_{0} = \lambda\int_{0}^{t}ds\, U_{t-s}^{\lambda}\lone^{(10)}P_{0}V_{s}^{\lambda}P_{0} + \lambda^{4/3}\int_{0}^{t}ds\, U_{t-s}^{\lambda}\lfour^{(10)}P_{0}V_{s}^{\lambda}P_{0}.
\end{equation*}

Putting $X_{t}^{\lambda} := P_{0}U_{t}^{\lambda}$, and introducing the new variables $\tau = \lambda^{2}t$ and $\sigma = \lambda^{2}u$ we get then
\begin{equation*}\begin{split}
W_{\lambda^{-2}\tau}^{\lambda} &= X_{\lambda^{-2}\tau}^{\lambda} + \int_{0}^{\tau}d\sigma X_{\lambda^{-2}(\tau - \sigma)}^{\lambda}K_{1, 1}(\lambda, \tau -\sigma)W_{\lambda^{-2}\sigma}^{\lambda}+\\
& + \lambda^{1/3}\int_{0}^{\tau}d\sigma X_{\lambda^{-2}(\tau - \sigma)}^{\lambda}[K_{1, 4/3}(\lambda, \tau -\sigma) + K_{4/3, 1}(\lambda, \tau -\sigma)]W_{\lambda^{-2}\sigma}^{\lambda}+\\
&+ \lambda^{2/3}\int_{0}^{\tau}d\sigma X_{\lambda^{-2}(\tau - \sigma)}^{\lambda}K_{4/3, 4/3}(\lambda, \tau -\sigma)W_{\lambda^{-2}\sigma}^{\lambda},
\end{split}
\end{equation*}
where 
\begin{equation*}
K_{i, j}(\lambda, \tau) := \int_{0}^{\lambda^{-2}\tau}dx\, X_{-x}^{\lambda}{\sf L}_{i}^{(01)}U_{x}^{\lambda}{\sf L}_{j}^{(10)}, \qquad i, j = 1, 4/3\quad .
\end{equation*}

One can show that, when $\lambda\to 0$, $K_{i, j}(\lambda, \tau)$ converges to
\begin{equation*}\label{kij}
K_{i, j} := \int_{0}^{\infty}dx\, {\sf L}_{i}^{(01)}U_{x}{\sf L}_{j}^{(10)},
\end{equation*}
so that
\begin{equation*}
W_{\lambda^{-2}\tau}^{\lambda} - \tilde{W}_{\lambda^{-2}\tau}^{\lambda} \to 0,
\end{equation*}
where 
\begin{equation*}
\tilde{W}_{t}^{\lambda} := \E^{-\I(\lambda^{2/3}\lp^{(00)} + \lambda^{2}K_{1, 1})t}\quad .
\end{equation*}

Spelling out the terms in $K_{1, 1}$ and formulating the result in $\hilbert{H}$ instead of expressing it in $\B$ one gets theorem \ref{ctoinfinity}.

\section*{Acknowledgments}
I am grateful to Stefan Teufel for numerous valuable discussions and remarks concerning this work and I thank the Deutsche Forschungsgemeinschaft (German Research Foundation) for financial support.

\end{document}